\documentclass[pdftex,envcountsame,runningheads]{llncs} 

\usepackage[latin1]{inputenc}
\usepackage{amsmath,amsfonts,amssymb,latexsym,eufrak}
\usepackage{stmaryrd,proof}
\usepackage[disable]{todonotes} %
\usepackage{hyperref}
\usepackage[inline]{enumitem}

\usepackage{xifthen}

\usepackage{tikz}
\usetikzlibrary{calc,arrows} 

\renewcommand{\subsubsection}[1]{\vspace{0.5ex}\noindent\textbf{{#1}}\,}

\providecommand*{\ifempty}[3]{\ifthenelse{\isempty{#1}}{#2}{#3}}
\newcommand{\parensmathoper}[2]{\ensuremath{#1\ifempty{#2}{}{(#2)}}}

\newcommand{\N}{\mathbb{N}}
\newcommand{\Q}[1][+]{\mathbb{Q}_{#1}}
\newcommand{\R}[1][+]{\mathbb{R}_{#1}}

\newcommand{\M}{\mathcal{M}}

\newcommand{\Exp}{\parensmathoper{\mathit{Exp}}}

\newcommand{\Distr}{\parensmathoper{\mathcal{D}}}

\newcommand{\stateat}[2]{#1 [ #2 ]}
\newcommand{\timeat}[2]{#1 \langle #2 \rangle}  
\newcommand{\prefix}[2]{ #1|^{#2}}
\newcommand{\tail}[2]{#1 |_{#2}} 

\newcommand{\cyl}{\parensmathoper{\mathfrak{C}}}
\newcommand{\paths}{ \parensmathoper{\Pi} }
\newcommand{\pathsC}{ \parensmathoper{\Pi^2}}
\newcommand{\sigmapaths}[1]{ \Sigma_{\paths{}\ifempty{#1}{}{(#1)}} }

\renewcommand{\Pr}[2][\M]{\mathbb{P}^{#1}_{#2}}

\newcommand{\false}{\bot}
\newcommand{\Next}[2][I]{\mathsf{X}^{#1} #2}
\newcommand{\Until}[3][I]{#2 \mathbin{\mathsf{U}^{#1}} #3}
\newcommand{\MTL}[1][]{\text{\normalfont MTL}^{#1}}

\newcommand{\DTA}{\mathrm{DTA}}
\newcommand{\sDTAR}{1\text{-}\mathrm{RDTA}}
\newcommand{\A}{\mathcal{A}}
\newcommand{\clk}{\mathcal{X}}
\newcommand{\clkV}{\mathcal{V}(\clk)}
\newcommand{\clkC}{\mathcal{G}(\clk)}
\newcommand{\orth}{\mathrel{\bot}}
\newcommand{\trs}[1]{\xrightarrow{\;#1\;}}
\renewcommand{\L}{\parensmathoper{\mathcal{L}}}

\newcommand{\cl}{\overline}
\newcommand{\eqlbl}{\mathrel{\equiv_\ell}}

\newcommand{\coupling}[2]{\Omega(#1, #2)}

\newcommand{\set}[2]{\left\{ #1 \ifempty{#2}{}{\mid #2} \right\}}
\newcommand{\pair}[2]{\langle #1,#2 \rangle}
\newcommand{\denot}[2][\M]{\llbracket #2 \rrbracket_{#1}}
\newcommand{\mDom}[1][S \times S]{[0,1]^{#1}}
\newcommand{\dist}[1][]{\delta_{\color{red}#1}}
\newcommand{\discr}[1]{\gamma^{\mathcal{#1}}}
\newcommand{\tv}[2]{\| #1 - #2 \|_{\text{\tiny TV}}}
\newcommand{\norm}[2][\infty]{\| #2 \|_{#1}}

\newcommand{\symdiff}{\mathbin{\triangle}}

\title{Topologies of Stochastic Markov Models: \\ Computational Aspects%
\thanks{Work supported by the Sino-Danish Basic Research Center IDEA4CPS, the EU Artemis M-BAT project, the EU FP7 CASSTING and SENSATION projects.}}

\author{Giorgio Bacci \and Giovanni Bacci \and Kim G.~Larsen \and Radu Mardare}
\institute{
Department of Computer Science, Aalborg University, Denmark \\
\email{\{grbacci,giovbacci,kgl,mardare\}@cs.aau.dk}
}

\authorrunning{G.\ Bacci, G.\ Bacci, K.\ Larsen, R.\ Mardare}

\begin{document}

\maketitle

\begin{abstract}
In this paper we propose two behavioral distances that support approximate reasoning on Stochastic Markov Models (SMMs), that are continuous-time stochastic transition systems where the residence time on each state is described by a generic probability measure on the positive real line. In particular, we study the problem of measuring the behavioral dissimilarity of two SMMs against linear real-time specifications expressed as Metric Temporal Logic (MTL) formulas or Deterministic Timed-Automata (DTA).

The most natural choice for such a distance is the one that measures the maximal difference that can be observed comparing two SMMs with respect to their probability of satisfying an arbitrary specification. We show that computing this metric is \textbf{NP}-hard. In addition, we show that any algorithm that approximates the distance within a certain absolute error, depending on the size of the SMMs, is \textbf{NP}-hard.

Nevertheless, we introduce an alternative distance, based on the Kantorovich metric, that is an over-approximation of the former and we show that, under mild assumptions on the residence time distributions, it can be computed in polynomial time.
\end{abstract}

\section{Introduction} \label{sec:intro}

Continuous-time probabilistic systems constitute the basic semantical tool to model random phenomena in complex real-time applications. They are successfully exploited in performance and dependability analysis 
and mostly used in applications such as systems biology, modeling/testing of cyber-physical systems, machine learning, and analysis of lossy network systems, etc.

In this context, the models are verified (or \emph{model checked}) against real-time specifications aiming at determining the probability by which these are attained.
In this way one may prove that unwanted behaviors are unlike to occur (e.g., with a low probability) within a given time horizon or that certain events happen according to a specific desired timing pattern with high probability.
Usually, real-time specifications are expressed as temporal logical formulas~\cite{Koymans90,AlurH93,AlurH94,AzizSSB00} or as the language recognized by automata-like formalisms such as Timed Automata (TAs)~\cite{AlurD94}. In this work we focus on the class of \emph{linear} real-time specifications expressed either as Metric Temporal Logic (MTL) formulas~\cite{Koymans90,AlurH93,AlurH94} or as the timed-languages recognized by Deterministic Timed Automata (DTAs)~\cite{AlurM04}.
Our attention on linear-time properties, opposed to branching-time properties, is motivated by the fact that in many applications the system to be modeled cannot be internally accessed, but only tested via observations performed over a set of random executions. For instance, this is mostly common in application domains such as systems biology, modeling/testing of cyber-physical systems, and machine learning.

However, when one aims at verifying properties of a real system out of its model representation, he should also take into consideration the degree of inaccuracy of the represention. Indeed, if the real-valued parameters of the model have been acquired from empirical data subject to error estimates, any analysis performed on it is itself subject to an inherent source of inaccuracy, that may lead to deceptive results on the original system. This motivated the study of \emph{behavioral metric semantics}, initially developed for discrete-time Markov chains (MCs)~\cite{DesharnaisGJP04}, then extended to continuous-time models, generalized semi-Markov processes~\cite{GuptaJP06}, and general Markov Processes~\cite{DesharnaisGJP04}, which provide for a formal notion of behavioral similarity between systems.

In this paper we study metrics aimed at helping the verification of continuos-time Markov systems against linear real-time specifications. In particular, we define two pseudometrics, respectively called \emph{MTL} and \emph{DTA} \emph{variation distances}, that measure the maximal difference that can be observed by comparing two models with respect to their likelihood of satisfying any specification expressed in the form of an MTL formula or as a timed language recognized by a DTA, respectively. Therefore, knowing that two models are at distance $\epsilon \geq 0$ from each other ensures that any result obtained by an analysis w.r.t.\ a specification (hence, an MTL formula or a DTA) on one system can be reflected to the other with an absolute error bounded by $\epsilon$.

\noindent The technical contributions of this paper can be summarized as follows.
\begin{enumerate}[topsep=0pt,noitemsep=1ex,fullwidth,itemindent=\parindent]
\item We introduce \emph{stochastic Markov models} (SMMs), that are continuous-time stochastic transition systems where the residence time on a state is specified by a generic probability measure on the positive real line. These generalize both MCs ---where probabilistic transitions happens instantaneously--- and time-ho\-mo\-ge\-ne\-ous Continuous-Time Markov Chains (CTMCs) ---where the residence time probability on states is characterized by a negative exponential distribution.
\item We study the measure-theoretical and topological properties of real-time specifications expressed as MTL formulas or DTAs. Specifically, we show that the $\sigma$-algebras generated, respectively, by MTL and DTA specifications, coincide. Then, we introduce a pseudometric between real-time specifications (that are, measurable sets of timed paths) and we single out a considerably simple family of specifications, namely those represented as resetting single-clock DTAs ($\sDTAR$), that is \emph{dense} in the whole $\sigma$-algebra. This will imply that the probability measured in any measurable set can be approximated arbitrarily close by an $\sDTAR$s. This has practical applications in quantitative model checking on CTMCs, since this allows one to exploit efficient algorithms for \emph{single clock} DTAs~\cite{KatoenLMCS11} to approximate model checking against DTAs or MTL formulas.

\item We consider two total variation distances on SMMs: one characterizing the maximal variation w.r.t. $\MTL$ specifications, the other w.r.t. $\DTA$ specifications. We show that these two pseudometrics coincide, and that the total variation can be obtained solely looking at $\sDTAR$s specifications.

\item We prove that the problem of computing the distance, \emph{exactly}, is \textbf{NP}-hard. This is done via a reduction from Max Clique, by generalizing an argument by Lyngs\o\ and Pedersen~\cite{LyngsoP02}. Furthermore, we show that even the problem of approximating the distance within a certain absolute error, depending on the size of the SMMs, is \textbf{NP}-hard. To the best of our knowledge, whether the distance is computable or not still remains an open problem.

\item Nevertheless, we provide a fixed point bisimilarity distance, based on the Kantorovich metric, that is an over-approximation of the former. Then, extending a result by Chen et al.~\cite{ChenBW12} for MCs, we prove that, under mild assumptions on the residence time distributions occurring in the SMMs, this distance can be computed in polynomial time using the ellipsoid algorithm.
\end{enumerate}

The paper has an appendix that contains some of the proofs that could not be included in the paper due to their size or complexity.

\section{Preliminaries} \label{sec:prelim}
In this section we recall the basic notions used in the paper and fix the notation. 

\subsubsection{Measure theory.} 
A \emph{field} over a set $X$ is a nonempty family $\mathcal{F}$ of subsets of $X$ closed under complement and finite union. A \emph{$\sigma$-algebra} $\Sigma$ over $X$ is a field also closed under countable union. The pair $(X, \Sigma)$ is called a \emph{measurable space} and the elements of $\Sigma$ \emph{measurable sets}. For a family $\mathcal{F}$ of subsets of $X$, the $\sigma$-algebra generated by $\mathcal F$ is the smallest $\sigma$-algebra containing $\mathcal{F}$, denote by $\sigma(\mathcal{F})$.

Consider two measurable spaces $(X, \Sigma_X)$ and $(Y, \Sigma_Y)$. A function $f \colon X \to Y$ is \emph{measurable} if for all $E \in \Sigma_Y$, $f^{-1}(E) = \{x \mid f(x) \in E\} \in \Sigma_X$. 
The \emph{product spaces}, $(X,\Sigma_Y) \otimes (Y, \Sigma_Y)$, is the measurable space $(X \times Y, \Sigma_X \otimes \Sigma_Y)$, where $\Sigma_X \otimes \Sigma_Y$ is the $\sigma$-algebra generated by the \emph{rectangles} $E \times F\in \Sigma_X\times\Sigma_Y$.

A \emph{measure} on a measurable space $(X, \Sigma)$ is a $\sigma$-additive function \mbox{$\mu \colon \Sigma \to [0,\infty]$}, i.e.~$\mu(\bigcup_{E \in \mathcal{F}} E) = \sum_{E \in \mathcal{F}} \mu(E)$ for all countable families $\mathcal{F}$ of pairwise disjoint measurable sets; it is a \emph{probability measure} if, in addition, $\mu(X) = 1$. We denote by $\Delta(X,\Sigma)$ the set of probability measures on $(X, \Sigma)$, and by $\Distr{X} = \Delta(X,2^X)$ the set of \emph{(discrete) probability distributions}. 

Given two measurable spaces $(X,\Sigma_X)$ and $(Y, \Sigma_Y)$, and a measurable function $f \colon X \to Y$, any measure $\mu$ on $(X,\Sigma_X)$ defines a measure $\mu \circ f^{-1}$ on $(Y,\Sigma_Y)$; this operation is called \emph{push forward}, denoted by $\mu\#f$. 

Given two measures $\mu$ and $\nu$, on $(X,\Sigma_X)$ and $(Y,\Sigma_Y)$, respectively, we define the \emph{product measure} $\mu \times \nu$ on $(X,\Sigma_X) \otimes (Y, \Sigma_Y)$, as the \emph{unique} measure such that $(\mu \times \nu)(E \times F) = \mu(E) \cdot \nu(E)$, for arbitrary $E \in \Sigma_X$ and $F \in \Sigma_Y$. A measure $\omega$ on $(X,\Sigma_X) \otimes (Y, \Sigma_Y)$ is a \emph{coupling} for $(\mu, \nu)$ if $\omega(E \times Y) = \mu(E)$ and $\omega(X \times F) = \nu(F)$, for arbitrary $E \in \Sigma_X$ and $F \in \Sigma_Y$; $\mu$ and $\nu$ are the \emph{left} and the \emph{right marginals} of $\omega$. We denote by $\coupling{\mu}{\nu}$ the set of couplings for $(\mu, \nu)$.

Throughout the paper $(\R, \Sigma_{\R})$, or simply $\R$, will denote the measurable space of positive real numbers with zero with Borel $\sigma$-algebra.

\subsubsection{Metric spaces.} 
Given a set $X$,  a function $d \colon X \times X \to \R$ is a \emph{pseudometric} on $X$ if $d(x,x) = 0$, $d(x,y) = d(y,x)$ and $d(x,y) + d(y,z) \geq d(x,z)$, for arbitrary $x,y,z \in X$; it is a \emph{metric} if, in addition, $d(x,y) = 0$ iff $x = y$. A pair $(X,d)$ where $d$ is a (pseudo)metric on $X$ is called a \emph{(pseudo)metric space}. 

Given a measurable space $(X,\Sigma)$, we consider two metrics on $\Delta(X,\Sigma)$: 
\begin{itemize}[label=\textbullet,itemsep=0.5ex,topsep=0pt]
\item the \emph{total variation distance}, defined for arbitrary $\mu, \nu \in \Delta(X,\Sigma)$ by
\begin{center}
\vspace{-1ex}
$\tv{\mu}{\nu} = \sup_{E \in \Sigma} |\mu(E)-\nu(E)|$,
\vspace{-1ex}
\end{center}
\item the \emph{Kantorovich (pseudo)metric}, defined for a (pseudo)metric $d$ on $X$ by
\begin{center}
\vspace{-1ex}
$\mathcal{K}_d(\mu, \nu) = \inf\set{ \int d\; \text{d}\omega}{\omega \in \coupling{\mu}{\nu}}$.
\vspace{-1ex}
\end{center}
\end{itemize}

\subsubsection{The space of timed paths.} 
A \emph{timed path} over a set $X$ is an alternating infinite sequence $\pi = x_0, t_0,x_1,t_1 \dots$ of elements $x_i \in X$ and \emph{time delays} \mbox{$t_{i} \in \R$}, for $i \in \N$. $\paths{X}$ denotes the set of timed paths over $X$. For arbitrary $i \in \N$, let $\stateat{\pi}{i} = x_i$, $\timeat{\pi}{i} = t_i$, $\prefix{\pi}{i} = x_0, t_0, \dots, t_{i-1}, x_i$ and $\tail{\pi}{i} = x_i, t_i, x_{i+1}, t_{i+1},\dots$.
For $X_i \subseteq X$, $R_i \subseteq \R$, $i=0..n$, let $\cyl{X_0, R_0, \dots, R_{n-1}, X_n}$ be the \emph{cylinder set} of the timed paths $\pi \in \paths{X}$ such that $\prefix{\pi}{n} \in X_0 \times R_0 \times \dots \times R_{n} \times X_n$.

For $(X,\Sigma)$ a measurable space, $\paths{X,\Sigma}$ is the measurable space of timed paths over $X$, with $\sigma$-algebra $\sigmapaths{X,\Sigma}$ generated by the \emph{measurable cylinders} $\cyl{S_0, R_0, \dots, R_{n-1}, S_n}$, where $S_i \in \Sigma$, $R_i \in \Sigma_{\R}$, $i=0..n$, and $n \in \N$. If $\Sigma = \sigma(\mathcal{F})$ and $\Sigma_{\R} = \sigma(\mathcal{I})$, then $\sigmapaths{X,\Sigma} = \sigma(\cyl{\mathcal{F}, \mathcal{I}})$, where $\cyl{\mathcal{F}, \mathcal{I}}$ is the family of cylinders $\cyl{F_0, I_0, \dots, F_{n-1}, I_n}$ where $F_i \in \mathcal{F}$, $I_i \in \mathcal{I}$, $i=0..n$, and $n \in \N$. Moreover, if both $\mathcal{F}$ and $\mathcal{I}$ are fields, so is $\cyl{\mathcal{F},\mathcal{I}}$.

For a function $f \colon X \to Y$, we define $f^\omega \colon \paths{X} \to \paths{Y}$ as the obvious stepwise extension of $f$ on timed paths. Note that if $f$ is measurable, so is $f^\omega$.

\section{Stochastic Markov Models} \label{sec:SMM}

In this section we introduce the class of \emph{Stochastic Markov Models} (SMMs) and define behavioral equivalences among them.
Let $AP$ be a countable set of \textit{atomic propositions}, that we fix for the rest of the paper.

\begin{definition}[Stochastic Markov Model] \label{def:SMM}
A \emph{stochastic Markov model} is a tuple $\M = (S, A, \tau, \rho, \ell)$ consisting of a finite nonempty set $S$ of \emph{states}, a set $A \subseteq S$ of \emph{absorbing states}, a \emph{transition probability function} $\tau \colon S \setminus A \to \Distr{S}$, an \emph{exit-time probability function} $\rho \colon S \setminus A \to \Delta(\R)$, and a \emph{labelling function} $\ell \colon S \to 2^{AP}$.
\end{definition}
The operational behavior of $\M = (S, A, \tau, \rho, \ell)$ can be described as follows: if the system is in state $s \in S$ and $s$ is absorbing, no transition can be made; otherwise, it moves to an arbitrary $s' \in S$ within time $t \in \R$ with probability $\rho(s)([0,t]) \cdot \tau(s)(s')$. An atomic proposition $p \in AP$ is said to hold in $s$ iff $p \in l(s)$.

SMMs subsume both MCs and  \emph{time-homogeneous} CTMCs. Indeed, MCs are the SMMs such that $A = \emptyset$ and, for all $s \in S$, $\rho(s)$ is the Dirac measure at $0$ (transitions happen instantaneously); CTMCs are the SMMs such that, for all $s \notin A$, $\rho(s) = \Exp\lambda$, where $\Exp\lambda$ denotes the negative exponential distribution with parameter $\lambda > 0$.

An SMM $\M = (S, A, \tau, \rho, \ell)$ induces an $S$-indexed family of probability measures on the measurable space $\paths{S,2^S}$ of timed paths over $S$ as follows.

\begin{definition} \label{def:SMMprob}
Let $\M = (S, A, \tau, \rho, \ell)$ be an SMM and $s \in S$\todo{restrict to non-absorbing}. The probability measure $\Pr{s}$ on the measurable space of timed paths over $S$ is the unique measure such that, for all $n \in \N$, $s_i \in S$ and $R_i \in \Sigma_{\R}$, $i=0..n+1$,
\begin{align*}
  \Pr{s}(\cyl{s_0}) &= \chi_{\{s\}}(s_0) \,,
  \\
  \Pr{s}(\cyl{s_0, R_0,\dots, R_n, s_{n+1}}) &= \Pr{s}(\cyl{s_0, R_0, \dots, R_{n-1}, s_n}) \cdot P(s_n, R_n, s_{n+1}) \,,
\end{align*}
where $\chi_{E}$ is the characteristic function of $E$ and, for $s,s' \in S$, $R \in \Sigma_{\R}$, $P(s, R, s') = \rho(s)(R) \cdot \tau(s)(s')$, if $s \notin A$, and $P(s, R, s') = 0$ otherwise. 
\end{definition}
The existence of this measure is guaranteed by the Hahn-Kolmogorov extension theorem; uniqueness is guaranteed since, for all $s \notin A$, $\tau(s)$ and $\rho(s)$ are $\sigma$-finite. Intuitively, $\Pr{s}$ describes the probability that a stochastic run of $\M$ starting from $s$ belongs to a measurable set of $\paths{S,2^S}$.

Next we introduce two important behavioral equivalences on SMMs: \emph{stochastic trace equivalence} and \emph{bisimilarity}. To do so, for an SMM $\M = (S, A, \tau, \rho, \ell)$, we first define the following equivalence relations on $S$:
\begin{itemize}
\item $s \eqlbl s'$ if and only if $\ell(s) = \ell(s')$;
\item $s \equiv_A s'$ if and only if, either $s,s' \in A$ or $s,s' \notin A$.
\item  $s \equiv s'$ if and only if $s \eqlbl s'$ and $s \equiv_A s'$.
\end{itemize}

For an SMM $\M = (S, A, \tau, \rho, \ell)$ we define $\mathcal{T}_\mathcal{M} = \cyl{S/_{\eqlbl}, \Sigma_{\R}}$ and its elements will be called \emph{trace cylinders}.
\begin{definition}[Stochastic Trace Equivalence]
Let $\M = (S, A, \tau, \rho, \ell)$ be an SMM. Two states $s, s' \in S$ are \emph{stochastic trace equivalent} with respect to $\M$, written $s \approx_\M s'$, if for all trace cylinders $T \in \mathcal{T}_\M$, $\Pr{s}(T) = \Pr{s'}(T)$.
\end{definition}
Clearly ${\approx_\M} \subseteq {\eqlbl}$, moreover since states are also tested with respect to their associated probability on timed paths, we also have ${\approx_\M} \subseteq {\equiv_A}$.

\begin{definition}[Bisimulation] \label{def:pbisim}
Let $\M = (S, A, \tau, \rho, \ell)$ be an SMM. An equivalence relation $R \subseteq S \times S$ is a \emph{bisimulation} on $\M$ if whenever $(s, s') \in R$, 
\begin{itemize}[topsep=0.5ex, noitemsep]
\item $s \equiv s'$; and 
\item if $s,s'\not\in A$, then $\rho(s) = \rho(s')$ and, for all $C \in S /_R$, $\tau(s)(C) = \tau(s')(C)$.
\end{itemize}

Two states $s, s' \in S$ are \emph{bisimilar} with respect to $\M$, written $s \sim_{\M} s'$, if they are related by some bisimulation on $\M$.
\end{definition}
These conditions require that any two bisimilar states are equally labelled and have identical probability w.r.t. any time delay of moving to any bisimilarity equivalence class. Note that, the above definition extends the notion of probabilistic bisimulation both on MCs~\cite{LarsenS91} and CTMCs in~\cite{BaierKHW05}.\todo{Say ${\approx_\M} \subsetneq {\sim_\M}$?}

\section{Pseudometrics for Linear Real-Time Specifications} \label{sec:pseudometrics}

We start by introducing the two kinds of linear real-time specifications that we consider throughout the paper.

\subsubsection{Metric Temporal Logic.}
Metric Temporal Logic ($\MTL$) \cite{AlurH93,AlurH94} has been introduced as a formalism for reasoning on sequences of events in a real-time setting. The grammar of logical formulas of $\MTL$ is as follows
\begin{equation*}
  \varphi ::= p \mid \false \mid \varphi \to \varphi \mid \Next[{[t,t']}]{\varphi} \mid \Until[{[t,t']}]{\varphi}{\varphi} \,,
\end{equation*}
where $p \in AP$ and $t, t' \in \Q$ with $t \leq t'$.

Following~\cite{OuaknineW07}, the semantics of $\MTL$ formulas $\varphi$ is given by means of a satisfiability relation $\M, \pi \models \varphi$, defined, for an SMM $\M = (S, A, \tau, \rho, \ell)$ and a timed path $\pi \in \paths{S}$, as follows.
\begin{align*}
  \M, \pi &\models p && \text{if } p \in \ell(\stateat{\pi}{0}) \,,
  \\
  \M, \pi &\models \false && \text{never} \,,
  \\
  \M, \pi &\models \varphi \to \psi && \text{if $\M, \pi \models \psi$ whenever $\M, \pi \models \varphi$} \,,
  \\
  \M, \pi &\models \Next[{[t,t']}]{\varphi} && \text{if   
  	$\timeat{\pi}{0} \in [t,t']$, and $ \M, \tail{\pi}{1} \models \varphi$} \,,
  \\
  \M, \pi &\models \Until[{[t,t']}]{\varphi}{\psi} && \text{if $\exists i > 0$ such that $\textstyle \sum_{k=0}^{i-1} \timeat{\pi}{i} \in [t,t']$, $\M, \tail{\pi}{i} \models \psi$, } \\
  	& && \text{\phantom{if}  
	and $\M, \tail{\pi}{j} \models \varphi$ whenever $0 \leq j < i$} 	\,.
\end{align*}
The above is usually referred to as the \emph{point-based} semantics. A key observation about this interpretation of formulas is that temporal connectives quantify over a countable set of positions in a timed path. In contrast, the \emph{interval-based} semantics, adopted e.g., in~\cite{ChenDKM11,SharmaK11}, associates a state to each point in real time, and the temporal connectives quantify over the whole timed domain.

For $\varphi \in \MTL$ we denote by $\denot{\varphi} = \set{ \pi \in \paths{S} }{ \M, \pi \models \varphi}$ the set of all timed paths satisfying $\varphi$ in $\M$, and define $\denot{\MTL} = \set{ \denot{\varphi} }{ \varphi \in \MTL }$.

The next lemma states that the sets of timed paths satisfying a given $\MTL$ formula are measurable in $\paths{S, 2^S}$. This justifies that $\MTL$ formulas can be used as linear real-time specifications for an SMMs.

\begin{lemma} \label{lem:MTLmeas}
Let $\M$ be an SMM, then $\denot{\MTL} \subseteq \Sigma_{\paths{S, 2^S}}$.
\end{lemma}

\subsubsection{Deterministic Timed Automata.} 
Timed Automata (TAs)~\cite{AlurD94} have been introduced to model the behavior of real-time systems over time. Here we consider their deterministic variant without location invariants.
 
Let $\clk$ be a finite set of $\R$-valued variables, called \emph{clocks}, and let $\clkV$ be the set of all valuations $v \colon \clk \to \R$ for the clocks in $\clk$. 
For $v \in \clkV$, $t \in \R$, and $X \subseteq \clk$, we denote by $\mathbf{0}$, the constant zero valuation, by $v + t$, the $t$-delay of $v$, and by $v[X := t]$, the update of $X$ in $v$, all defined in the obvious way. 

A \emph{clock guard} $g \in \clkC$ over $\clk$ is a finite set of expressions of the form $x \bowtie q$, for $x \in \clk$, $q \in \Q$, and ${\bowtie} \in \set{{<}, {\leq}, {>}, {\geq}}{}$. We say that a valuation $v \in \clkV$ \emph{satisfies} a clock guard $g \in \clkC$, written $v \models g$, if $v(x) \bowtie n$ holds, for all $x \bowtie q \in g$; two clock guards $g, g' \in \clkC$ are \emph{orthogonal} (or \emph{non-overlapping}), written $g \orth g'$, if there is no $v \in \clkV$ such that $v \models g$ and $v \models g'$.

\begin{definition}[Deterministic Timed Automata] \label{def:dta}
A \emph{deterministic timed automaton} over a set of clocks $\clk$ is a tuple $\A = (Q, L, q_0, F, \to)$ consisting of a finite set $Q$ of \emph{locations}, a set $L$ of \emph{symbols}, an \emph{initial location} $q_0 \in Q$, a set $F \subseteq Q$ of \emph{final locations}, and a \emph{transition relation} ${\to} \subseteq Q \times L \times \clkC \times 2^\clk \times Q$ such that, whenever $(q,a,g,X,q'), (q,a,g',X',q'') \in {\to}$ and $g \neq g'$, then $g \orth g'$.
\end{definition} 
An \emph{run} of a DTA $\A = (Q, L, q_0, F, \to)$ over a timed path $\pi = a_0,t_0,a_1,t_1,\ldots$ over $L$, is and infinite sequence of the form
\begin{equation*}
  r = (q_0, v_0) \trs{a_0, t_0} (q_1, v_1) \trs{a_1, t_1} (q_2, v_2) \trs{a_2, t_2} \cdots
\end{equation*}
with $q_i \in Q$ and $v_i \in \clkV$, for all $i \geq 0$, satisfying the following requirements: 
(\emph{initialization}) $v_0 = \mathbf{0}$; (\emph{consecution}) for all $i \geq 0$, $v_{i+1} = (v_i + t_i)[X_i := 0]$, for some $(q_i, a_i, g_i, X_i, q_{i+1}) \in {\to}$ such that $v_i + t_i \models g_i$. 

A run as above is \emph{accepting} if $q_i \in F$, for some $i \geq 0$, and we say that $\pi$ is \emph{accepted} by $\A$. $\L{\A}$ denotes the collection of all timed paths accepted by $\A$.

Observe that, due to the condition imposed on the transition relation, a deterministic timed automaton has at most one accepting run over a given timed path in $\paths{L}$. Moreover, differently from TAs, which are only closed under finite union and intersection, DTAs are also closed under complement~\cite{AlurD94}.

Following~\cite{KatoenLMCS11}, a DTA accepting input symbols in $2^{AP}$ can be thought of as a linear real-time specification for SMMs. Formally, let $\DTA(L)$ denote the collection of DTAs accepting symbols in $L$, then, for $\A \in \DTA(2^{AP})$ and an SMM $\M = (S, A, \tau, \rho, \ell)$, we define $\denot{\A} = \set{\pi \in \paths{S}}{ \ell^\omega(\pi) \in \L{\A}}$ as the set of all timed paths in $\M$ accepted by $\A$, and $\denot{\DTA} = \set{\denot{\A}}{ \A \in \DTA(2^{AP})}$.

The next lemma justifies the use of DTAs as specifications for SMMs.
\begin{lemma} \label{lem:DTAmeas}
Let $\M$ be an SMM, then $\denot{\DTA} \subseteq \Sigma_{\paths{S, 2^S}}$.
\end{lemma}

In the rest of the paper $\DTA$s will be used only as specification on SMMs, so that $\DTA(2^{AP})$ will be simply denoted by $\DTA$.

\subsection{Topological properties of MTL and DTA specifications}

In this section we analyze the measure-theoretical and topological properties of $\MTL$ and $\DTA$ specifications.

The next lemma states that $\MTL$ and $\DTA$ specifications generate the same $\sigma$-algebra. Intuitively, two type of linear real-time specifications we consider, after being completed under complement and countable union, have the same expressivity. Moreover,  the same $\sigma$-algebra can be generated by trace cylinders.
\begin{lemma} \label{lem:sigmaEquiv}
Let $\M$ be an SMM. Then $\sigma(\denot{\MTL}) = \sigma(\denot{\DTA}) = \sigma(\mathcal{T}_\M)$.
\end{lemma}
From now on the $\sigma$-algebras in Lemma~\ref{lem:sigmaEquiv} will be simply referred as $\Sigma_\M$.

Now we consider the topological properties of $\Sigma_\M$. Let $(X, \Sigma)$ be a measurable space, then any measure $\mu$ over it induces a pseudometric $d_\mu \colon \Sigma \times \Sigma \to \R$ on $\Sigma$, a.k.a.\ the \emph{Fr\'echet-Nikodym pseudometric w.r.t. $\mu$}, defined, for $E, F \in \Sigma$, by $d_\mu(E,F) = \mu(E \symdiff F)$, where $\symdiff$ is the symmetric difference between sets\footnote{Triangular inequality follows by monotonicity and sub-additivity of $\mu$ noticing that, $A \symdiff C \subseteq (A \symdiff B) \cup (B \symdiff C)$}.

\begin{lemma} \label{lem:densefield}
Let $(X, \Sigma)$ be a measurable space and $\mu$ be a finite measure on it. If $\Sigma$ is generated by a field $\mathcal{F}$, then $\mathcal{F}$ is dense in the pseudometric space $(\Sigma, d_\mu)$.
\end{lemma}
Note that Lemma~\ref{lem:densefield} is generic both in the field and in the measure that are given. In particular, since $\DTA$s are closed under all Boolean operations~\cite{AlurD94}, $\denot{\DTA}$ forms a field of sets; the same holds for $\denot{\MTL}$. Hence, we have the following.
\begin{corollary} \label{cor:denseMTLDTA}
Let $M = (S,A,\tau,\rho,\ell)$ be an SMM, and $s \in S$. Then, $\denot{\MTL}$ and $\denot{\DTA}$ are dense in $(\Sigma_\M, d_{\Pr{s}})$.
\end{corollary}

\subsubsection{Single-clock Resetting DTAs.} 
The problem of model checking CTMCs a\-gainst TA specifications is known to been computationally very hard\todo{undecidable?}, even restricting to the subclass of DTAs. Recently, Chen et al.~\cite{KatoenLMCS11} provided an algorithm that is efficient for \emph{single-clock} DTAs. In this view, we show that the subclass of \emph{resetting single-clock} DTAs ($\sDTAR$s) (i.e., DTAs with a single clock that is reset whenever a transition to the next location occurs) can be used to approximate with arbitrary precision any DTA or MTL specification. 
Indeed, $\sDTAR$s are closed under Boolean operations\footnote{Closure under union follows by the standard product construction, noticing that duplications of clocks are needed only to ensure the right reseting of the clocks.}, thus to use Lemma~\ref{lem:densefield} it only remains to show the following.

\begin{lemma}  \label{lem:sDTARfield}
$\denot{\sDTAR}$ is a generator for $\Sigma_\M$.
\end{lemma}
 
\begin{theorem}[$\sDTAR$-approximant] \label{th:1RDTAapproximant}
Let $\M = (S, A, \tau, \rho, \ell)$ be an SMM and $s \in S$. Then, for any $E \in \Sigma_\M$ and any $\epsilon > 0$, there exists $\A \in \sDTAR$ such that $|\Pr{s}(E) - \Pr{s}(\denot{\A})| < \epsilon$.
\end{theorem}
Theorem~\ref{th:1RDTAapproximant} allows one to use the algorithm of~\cite{KatoenLMCS11}, to deploy approximate model checking of CTMCs against any real-time specification that is measurable in $\Sigma_\M$. We recall that this is the case both for DTAs and MTL formulas.

\subsection{MTL and DTA Variation Pseudometrics}

We consider distances on SMMs, specifically, variation pseudometrics parametric on the family of specifications where the maximal difference is meant to be tested.

\begin{definition}[Variation distance] \label{def:vardist}
Let $\M = (S, A, \tau, \rho, \ell)$ be an SMM and $\mathcal{F} \subseteq \Sigma_{\paths{S, 2^S}}$. We define the \emph{$\mathcal{F}$-variation pseudometric} $\dist^\M_{\mathcal{F}} \colon S \times S \to [0,1]$ as 
\begin{equation*}
  \dist^\M_{\mathcal{F}}(s,s') = \textstyle\sup_{E \in \mathcal{F}} | \Pr{s}(E) - \Pr{s'}(E) | \,.
\end{equation*}
\end{definition}

Lemmas~\ref{lem:MTLmeas} and \ref{lem:DTAmeas} justify to consider the variation pseudometric w.r.t.\ $\MTL$ and $\DTA$ specifications, i.e., $\dist^\M_{\MTL}$ and $\dist^\M_{\DTA}$, respectively%
\footnote{Formally, the two distances should be denoted as $\dist^\M_{\denot{\MTL}}$ and $\dist^\M_{\denot{\DTA}}$, however the simplified notation will not cause any problem.}. 
In particular, as we have already done in the case of specifications (Lemma~\ref{lem:sigmaEquiv}) we would like to compare the expressivity of the two metrics.
  
\begin{lemma} \label{lem:vardistonfield}
Let $(X, \Sigma)$ be a measurable space and $\mu$, $\nu$ be two finite measures on it. If $\Sigma$ is generated by a field $\mathcal{F}$, then $\tv{\mu}{\nu} = \sup_{E \in \mathcal{F}} |\mu(E) - \nu(E)|$.
\end{lemma}

At this point we can state the main result of this section.
\begin{theorem} \label{thm:equiDist}
Let $\M$ be an SMM. Then $\dist^\M_{\MTL} = \dist^\M_{\DTA} = \dist^\M_{\Sigma_\M}$.
\end{theorem}

Now we will study the property of $\dist^\M_{\Sigma_\M}$, which will be simply referred as $\dist^\M$ (or $\dist$) in the following. The first property, is that $\dist^\M$ is actually a behavioral distance in the sense that its kernel coincide with stochastic trace equivalence.
\begin{theorem} \label{th:tracedist}
Let $\M$ be an SMM. Then ${\approx_\M} = \set{ (s,s') }{\dist^\M(s,s') = 0 }$.
\end{theorem}

The next corollary is an immediate consequence of Lemmas~\ref{lem:vardistonfield} and \ref{lem:sDTARfield} noticing  that $\sDTAR$s are closed under all Boolean operations.
\begin{corollary}
Let $\M$ be an SMM, then $\dist^\M = \dist^\M_{\sDTAR}$.
\end{corollary}
This result is quite important, since it means that one can \emph{exactly} determine the variational distance w.r.t.\ $\MTL$ and $\DTA$ specifications only looking at the subclass of $\sDTAR$ specifications, for which we already observed that quantitative model checking problem admits efficient computational solutions~\cite{KatoenLMCS11}. 

\section{NP-Hardness and Inapproximability}
In this section we show that computing $\dist$ is \textbf{NP}-hard. In addition, we prove that, for some $\epsilon$ depending on the size of the model, even the problem of approximating $\dist$ within an absolute error $\epsilon$ is \textbf{NP}-hard. 

To this end we identify a subclass of SMMs where the total variation distance is characterized in terms of an $L_1$ distance over a suitable sub-$\sigma$-algebra of $\Sigma_\M$, generated by a family of cylinders in $\mathcal{T}_\M$ that we called \emph{word cylinders sets}.
\begin{definition}[Word cylinders]
Let $\M = (S, A, \tau, \rho, \ell)$ be an SMM. We define $\mathcal{W}_\M = \cyl{S/_{\eqlbl}, \set{\R}{}}$,
and its elements \emph{words cylinder}.
\end{definition}
Note that, the word cylinder sets are pairwise disjoint and, since the set of states is assumed to be finite, 
$\mathcal{W}_\M$ has countably many elements. This means that any $E \in \sigma(\mathcal{W}_\M)$ can be expressed as a countable union of word cylinder sets.

Under the assumption that the residence time distributions that occur in $\M$ are all equal, we can characterize $\dist^\M$ in terms of the $L_1$ distance between probability distributions in the measurable space $(\paths{S}, \sigma(\mathcal{W}_\M))$\footnote{We recall that for $\nu, \mu \in \Distr{X}$, 
$L_p(\mu, \nu) = \big(\sum_{x \in X} |\mu(x) - \nu(x)|^{p}\big)^{1/p}$.}.
\begin{lemma} \label{lem:LpDist}
Let $\M = (S, A, \tau, \rho, \ell)$ be an SMM such that $\rho(s) = \rho(s')$ for all $s, s' \not\in A$, then
$2 \cdot \dist^\M(s,s') = \textstyle\sum_{ E \in \mathcal{W}_\mathcal{M} } {|\Pr{s}(E) - \Pr{s'}(E)|}$.
\end{lemma}

The previous correspondence allows us to prove that computing $\dist$ on generic SMMs is \textbf{NP}-hard. 
The proof is carried out following an argument similar to Lyngs\o\ and Pedersen~\cite{LyngsoP02}, who proved the \textbf{NP}-hardness of comparing hidden Markov models (HMM) w.r.t.\ the $L_1$ norm. Namely, we show that the size of the maximum clique in an undirected graph with $n$ vertices can be computed within the time it takes to solve a Toeplix system\footnote{A Toeplix system is a linear system of equations where the coefficient matrix has each descending diagonal from left to right constant. Toeplix systems can be efficiently solved in time $\Theta(n^2)$ using the Levinson-Durbin procedure.} with $n$ unknowns and constant terms obtained from computing the distance $\dist$ for some SMMs that can be constructed in polynomial time in the size\footnote{We denote by $\mathit{size}(X)$ the representation of an object $X$. In particular, rational numbers are represented as quotient of integers written in binary.} of the input graph.

\begin{theorem}[NP-hardness] \label{th:NPhardness}
Computing $\dist$ over SMMs is \textbf{NP}-hard.
\end{theorem}

The proof of Theorem~\ref{th:NPhardness} makes use of SMMs that share the same residence time distribution on each state and it is generic in this choice\footnote{The only restriction consists in its representation, that has to be such that the construction made in Theorem~\ref{th:NPhardness} is polynomial in the size of the starting graph.}. This implies that the hardness result holds also within MCs or CTMCs.
\begin{corollary}
Computing $\delta$ over MCs or CTMCs is \textbf{NP}-hard.
\end{corollary}

Recently, Cortes et al.~\cite{CortesMR07} proposed a reduction similar to that of~\cite{LyngsoP02} to prove that computing the $L_{2p+1}$, for any $p \in \N$, between probabilistic automata is \textbf{NP}-hard to approximate within some absolute error that depends on the size of the given automata. Here we slightly generalize this idea on SMMs.

To this end, we first introduce some notation. 
Let $P \colon X \to \R[]$, we say that the algorithm $A$ approximates $P$, respectively
\begin{itemize}[topsep=0.5ex, noitemsep]
\item within an absolute error $\epsilon > 0$, if for all $\forall x \in X.\, |A(w) - P(w)| \leq \epsilon$; or
\item within a (multiplicative) factor $\alpha > 1$, if $\forall x \in X.\,P(w) \leq A(w) \leq \alpha \cdot P(w)$.
\end{itemize}
\begin{proposition}\label{prp:inapproximability}
Let $\M$ be an SMM, and $\alpha > 1$. If there exists a polynomial-time algorithm that approximates $\delta^\M$ within an absolute error $\epsilon = f(\alpha,size(\M))$, for some $f \colon \R[] \times \N \to [0,1]$, then there exists a polynomial-time algorithm that approximates Max Clique within a factor $\alpha$.
\end{proposition}
Proposition~\ref{prp:inapproximability} makes any inapproximability result for Max Clique to turn into an inapproximability result for the problem of computing $\delta$.
A famous result by Feige et al.~\cite{FeigeGLSS96} states that any algorithm that approximates Max Clique within any constant factor is \textbf{NP}-hard. As a corollary of that, we get the following.
\begin{corollary}[Inapproximability]
Given an SMM $\M$, there exists $\epsilon \in (0,1]$, depending on the size of $\M$, such that any algorithm that approximates $\delta^\M$ within an absolute error $\epsilon$ is \textbf{NP}-hard.
\end{corollary}

\section{A Polynomially Computable Upper-Bound}

In the literature, the problem of estimating the total variation distance is approached from two different perspectives. One consists in computing sharp estimates on suitable classes of distributions~\cite{HerbeiK13}; the second leverages on a well known relation between the total variation and the Kantorovic metric.
In this work we pursue the second way, leading to the definition of an over-approximation of $\delta$ that enjoys two good properties: 
\begin{enumerate*}[label=\emph{(\alph*)}]
\item it is computable in polynomial time, and
\item it is a pseudometric.
\end{enumerate*}
In particular, the latter property allows one to translate any convergence obtained with the over-approximation into a convergence with respect to the trace pseudometric.

This goal is achieved in three steps:
\begin{enumerate*}[label=\emph{(\roman*)}]
\item we construct an over-approximation based on the notion of coupling model;
\item we give a fixed point characterization of this over-approximation showing that it is a Kantorivich-based bisimilarity distance;
\item we show that, under mild assumptions on the residence time distributions, the fixed point is computable in polynomial time.
\end{enumerate*}
\subsection{An Over-Approximation}

\subsubsection{Coupling characterzation.} 
The construction of the over-approximation of $\dist$ will be based on a classic duality between the total variation distance of probability measures and their couplings~(see \cite[Theorem~5.2]{LindvallBook}), which states that: \textit{the maximal variation in the probabilities, evaluated among all the measurable sets, corresponds to the minimal discrepancy measured among all the possible couplings.}
Formally, the discrepancy associated to a coupling $\omega \in \coupling{\mu}{\nu}$ for two measures $\mu, \nu \in \Delta(X,\Sigma)$, is the value $\omega(\neq)$ associated to the measurable set ${\neq} \in \Sigma \otimes \Sigma$, where ${\neq} = \set{(x,y)}{x \neq y \in X}$. Since $\dist$ is a relaxation of the total variation distance, also the notion of discrepancy is relaxed accordingly. 
\newcommand{\unsep}{\mathrel{\eqlbl^\omega}}
\newcommand{\sep}{\not\unsep}
\begin{definition}
Let $\M = (S, A, \tau, \rho, \ell)$ be an SMM. We define ${\unsep} \subseteq \paths{S} \times \paths{S}$ as $\pi \unsep \pi'$ if and only if $\ell^\omega(\pi) = \ell^\omega(\pi')$.
\end{definition}
Intuitively, $\unsep$ is the stepwise extension of $\eqlbl$ to timed paths in the following sense: $\pi \unsep \pi'$ if $\stateat{\pi}{i} \eqlbl \stateat{\pi'}{i}$ and $\timeat{\pi}{i} =\ell \timeat{\pi'}{i}$ for all $i \in \N$. 

Lemma~\ref{lem:separability} and \ref{lem:congmeasurable} justify $\sep$ as the adequate notion for separability in $\Sigma_\M$.
\begin{lemma}[Separability]\label{lem:separability}
Let $\M \mathbin{=} (S, A, \tau, \rho, \ell)$ be an SMM and $\pi,\pi' \mathbin{\in} \paths{S}$. Then, $\pi \unsep \pi'$ iff $[ \text{for all } C \in \Sigma_\M, \, \pi \in C \text{ iff } \pi' \in C]$.
\end{lemma}
\begin{lemma} \label{lem:congmeasurable}
Let $\M = (S, A, \tau, \rho, \ell)$ be an SMM, then ${\sep} \in \Sigma_\M \otimes \Sigma_\M$.
\end{lemma}

The behavioral analogue of the duality between the total variation distance between probability measures and their couplings is generalized in $\Sigma_\M$ as follows.
\begin{lemma} \label{lem:dualityvariationandcoupling}
Let $\M = (S, A, \tau, \rho, \ell)$ be an SMM and $s,s' \in S$, then
\begin{equation*}
  \textstyle
  \sup_{E \in \Sigma_\M} |\Pr{s}(E) - \Pr{s'}(E)| = 
  \min \set{ \omega(\sep) }{ \omega \in \coupling{\Pr{s}}{\Pr{s'}} } \,.
\end{equation*}
\end{lemma}
The characterization given by Lemma~\ref{lem:dualityvariationandcoupling} suggests us to over-approximate $\delta$ by restricting the couplings where to search for the minimal discrepancy.
\newcommand{\C}{\mathcal{C}}
\begin{definition}[Coupling models] \label{def:coupling}
Let $\M = (S, A, \tau, \rho, \ell)$ be an SMM.  A \emph{coupling model for $\M$} is a tuple $\C = (S', A', \tau', \rho', \ell')$ such that 
\begin{enumerate}[leftmargin=*, label={\roman*.}, topsep=0.5ex, itemsep=0.5ex]
  \item $S' = S \times S$ and $A' = (A \times S) \cup (S \times A)$; 
  \item $\tau' \colon S' \setminus A' \to \Distr{S'}$ and for all $s,s' \notin A$, $\tau'(s,s') \in \coupling{\tau(s)}{\tau(s')}$;
  \item $\rho' \colon S' \setminus A' \to \Delta(\R \times \R)$ and for all $s,s' \notin A$, $\rho'(s,s') \in \coupling{\rho(s)}{\rho(s')}$;
  \item $\ell' \colon S' \to 2^{AP} \times 2^{AP}$ and for all $s,s' \in S$, $\ell'(s,s') = (\ell(s),\ell(s'))$.
\end{enumerate}
\end{definition}
A coupling model can be seen as a probabilistic pairing of two copies of $\M$ running synchronously.
The runs generated by a coupling model are \emph{coupled timed paths} of the form $(s_0,s'_0),(t_0,t'_0),(s_1,s'_1),(t_1,t'_1),\dots$, i.e., infinite alternating sequences over $(S \times S)$ and $(\R \times \R)$. For the sake of readability, we denote by $\pair{\pi}{\pi'}$ a coupled timed path pairing $\pi = s_0,t_0,s_1,t_1,\dots$ and $\pi' = s'_0,t'_0,s'_1,t'_1,\dots$; and by $\pathsC{S}$ the set of coupled timed paths over $S$.

The measurable space $\pathsC{S,2^S}$ of coupled timed paths over $S$ has $\pathsC{S}$ as underlying set and $\sigma$-algebra $\Sigma_{\pathsC{S}} = \sigma(\cyl{2^{S {\times} S}, \Sigma_{\R} \otimes \Sigma_{\R}})$. Couplings induce probability measures over the space of coupled timed paths as follows.
\begin{definition} \label{def:SMMprob}
Let $\M = (S, A, \tau, \rho, \ell)$ be an SMM, $\C = (S', A', \tau', \rho', \ell')$ be a coupling for it, and $s' \in S'$. The probability measure $\Pr[\C]{s'}$ on the measurable space of coupled timed paths over $S$ is the unique measure such that, for all $n \in \N$, $s'_i \in S'$ and $R'_i \in \Sigma_{\R} \otimes \Sigma_{\R}$, $i=0..n+1$
\begin{align*}
  \Pr[\C]{s'}(\cyl{s'_0}) &= \chi_{\{s'\}}(s'_0) \,,
  \\
  \Pr[\C]{s'}(\cyl{s'_0, R'_0,\dots,s'_n, R'_n, s'_{n+1}}) &= \Pr[\C]{s'}(\cyl{s'_0, R'_0, \dots, s'_n}) \cdot P'(s'_n, R'_n, s'_{n+1}) \,,
\end{align*}
where, for $u',v' \in S'$ and $R' \in \Sigma_{\R} \otimes \Sigma_{\R}$, $P'(u', R', v') = \rho'(u')(R') \cdot \tau'(u')(v')$, if $u' \notin A'$, and $P(u', R', v') = 0$ otherwise.
\end{definition}
The existence and the uniqueness of this measure follows from Hahn-Kolmogorov extension theorem \cite{LindvallBook}.

The name ``coupling model'' is justified by the fact that any measure $\Pr[\C]{s,s'}$ corresponds to a coupling for $(\Pr{s},\Pr{s'})$. This correspondence is obtained as a push forward w.r.t.\ the measurable function $\eta \colon \pathsC{S} \to \paths{S} \times \paths{S}$ that assigns to a coupled timed path $\pair{\pi}{\pi'}$ its (decoupled) pair $(\pi,\pi')$.
\begin{lemma} \label{lem:couplCorrespondence}
Let $\M = (S, A, \tau, \rho, \ell)$ be an SMM, $\C$ an arbitrary coupling model for $\M$ and $s, s' \in S$, then
$\Pr[\C]{s,s'}\#\eta \in \coupling{\Pr{s}}{\Pr{s'}}$.
\end{lemma} 
Let $\M = (S, A, \tau, \rho, \ell)$ be an SMM, we define $\vartheta^\M \colon S \times S \to [0,1]$ as
\begin{equation}
\vartheta^\M(s,s') = \min \set{(\Pr[\C]{s,s'}\#\eta)(\sep)}{\text{$\C$ coupling for $\M$}} \,.
\end{equation} 
By Lemmas~\ref{lem:couplCorrespondence} and \ref{lem:dualityvariationandcoupling} we obtain that $\vartheta^\M$ is an over-approximation of $\delta^\M$.
\begin{theorem}\label{th:CouplingCharacterization}
Let  $s,s'$ be states of the SMM $\M$. Then, $\dist^\M(s,s') \leq \vartheta^\M(s,s')$.
\end{theorem}

\subsubsection{Fixed point characterization.} \label{sec:fixpoint}
In this section we give a fixed point characterization of $\vartheta^\M$ based on the Kantorovich metric and the  total variation distance.

Given a coupling model $\C$ for the SMM $\M$, and a pair of states $s,s'$ of $\M$, the discrepancy between $s$ and $s'$ on $\C$, namely $(\Pr[\C]{s,s'}\#\eta)(\sep)$, can be interpreted as a reachability probability. Indeed, for any pair of timed paths $\pi,\pi' \in \paths{S}$, that starts in $s$ and $s'$ respectively, the probability associated with the event $\pi \sep \pi'$ corresponds to the probability of reaching a point $k \in \N$, where $\stateat{\pi}{k} \not\equiv \stateat{\pi'}{k}$ or $\timeat{\pi}{k} \neq \timeat{\pi'}{k}$. In the spirit of the fixed point characterization of reachability probabilities in the case of MCs~\cite[Theorem 10.15]{BaierK08}, we give a fixed point characterization of the discrepancy between states.

Consider the set of $[0,1]$-valued functions on $S \times S$, denoted by $\mDom$, endowed with the partial order $\sqsubseteq$ defined by $d \sqsubseteq d'$ iff $d(s,s') \leq d'(s,s')$ for all $s,s' \in S$. This forms a complete lattice with bottom ${\bf 0}$ and top ${\bf 1}$, defined as ${\bf 0}(s,s') = 0$ and ${\bf 1}(s,s') = 1$, for all $s,s' \in S$. For $D \subseteq \mDom$, the least upper bound $\bigsqcup D$, and greatest lower bound $\bigsqcap D$ are, respectively, given as $(\bigsqcup D)(s,t) = \sup_{d \in D} d(s,t)$ and $(\bigsqcap D)(s,t) = \inf_{d \in D} d(s,t)$, for all $s,t \in S$.

Consider an SMM $\M = (S, A, \tau, \rho, \ell)$ and a coupling $\C = (S', A', \tau', \rho', \ell')$ for $\M$. We define $\Gamma^\C \colon \mDom \to \mDom$ for $d \colon S \times S \to [0,1]$ by
\begin{equation*}
\Gamma^\C(d)(s,s') =
\begin{cases}
     1 & \text{if $s \not\equiv s'$} \\
     0 & \text{if $s \equiv s'$, and $s,s' \in A$} \\
     \alpha + (1 - \alpha) \displaystyle \sum_{u,v \in S} d(u,v) \, \tau'(s,s')(u,v) 
     &\begin{aligned}[t]
      & \text{otherwise} \\  
      & \text{for $\alpha = \rho'(s,s')(\neq_{\R})$}
     \end{aligned}
\end{cases}
\end{equation*}
where $\neq_{\R}$ denotes the measurable set $\{(x,y) \in \R^2 \mid x \neq y\} \in \Sigma_{\R} \otimes \Sigma_{\R}$.

One can easily verify that $\Gamma^\C$ is well-defined and order preserving, so that, by Tarski's fixed point theorem, $\Gamma^\C$ admits a least fixed point, denoted by $\discr{C}$.

The lemma below states that $\discr{C}$ corresponds to  the discrepancy w.r.t.\ $\C$.
\begin{lemma} \label{lemma:CouplingReachability}
Let $\mathcal{C}$ be a coupling for an SMM $\M = (S, A, \tau, \rho, \ell)$, and $s,s' \in S$. Then, $(\Pr[\C]{s,s'}\#\eta)(\sep) = \discr{C}(s,s')$.
\end{lemma}
As an immediate corollary we obtain the following characterization of $\vartheta^\M$.
\begin{corollary}\label{cor:modelstheta}
Let $\mathcal{M}$ be an SMM.
Then, $\vartheta^\M = \min \set{\discr{C}}{\text{$\mathcal{C}$ coupling for $\mathcal{M}$}}$.
\end{corollary}

In \cite{ChenBW12,BacciLM:tacas13} it is given a coupling-based characterization of the bisimilarity distance of Desharnais et al.~\cite{DesharnaisGJP04} for MCs, starting from the fixed point characterization by van Breugel et al.~\cite{BreugelSW08}. 
Here we follow the same idea but backwards: we start from Corollary~\ref{cor:modelstheta} to give a fixed point characterization of $\vartheta^\M$. 

The following fixed point operator will do the job.

For $\M = (S, A, \tau, \rho, \ell)$ a SMM, the operator $F^{\M} \colon \mDom \to \mDom$, 
for $d \colon S \times S \to [0,1]$ and $s,s' \in S$, is defined by:
\begin{equation*}
    F^{\M}(d)(s,s') = 
    \begin{cases}
        1 &\text{if $s \not\equiv s'$} \\
    	0 &\text{if $s \equiv s'$ and $s,s' \in A$} \\
	\alpha + (1 - \alpha) \, \mathcal{K}_d(\tau(s), \tau(s'))
     	&\begin{aligned}[t]
      	& \text{otherwise} \\  
      	& \text{for $\alpha = \tv{\rho(s)}{\rho(s')}$}
    	 \end{aligned}
    \end{cases}
\end{equation*}
where, for arbitrary $d \colon S \times S \to [0,1]$ and $\mu,\nu \in \Distr{S}$, $\mathcal{K}_d$ is the Kantorovich pseudometric that equals to $\mathcal{K}_d(\mu, \nu) = \min_{\omega \in \coupling{\mu}{\nu}} \sum_{u, v \in S} \omega(u,v) \cdot d(u,v)$ in the finite discrete case\footnote{
Since $S$ is finite, $\coupling{\mu}{\nu}$ describes a \emph{bounded} transportation polytope, hence the minimum in the definition of $\mathcal{K}_d(\mu, \nu)$ exists and can be achieved at some vertex.}.

$F^\M$ is easily seen to be monotonic, thus,
by Tarski's fixed point theorem, it has a least fixed point, that we denote by $f^\M$.

The next lemma states that the discrepancy measured w.r.t.\ arbitrary couplings $\C$ for $\M$ is an over-approximation of $f^\M$.
\begin{lemma} \label{lem:post-fixed}
Let $\mathcal{C}$ be a coupling for an SMM $\mathcal{M}$. If $d = \Gamma^{\C}(d)$, then $f^\M \sqsubseteq d$.
\end{lemma}
Lemma~\ref{lem:post-fixed} and Corollary~\ref{cor:modelstheta} yield the following characterization for $\vartheta^\M$. 
\begin{lemma} \label{lem:coupling-dist}
Let $\mathcal{M}$ be an SMM. Then, $f^\M = \vartheta^\M$.
\end{lemma}

By Lemma~\ref{lem:coupling-dist} we can finally prove that $\vartheta^\M$ is a pseudometric for SMMs. Furthermore, we show that $\vartheta^\M$ is a bisimilarity pseudometric in the sense of~\cite{FernsPP04}.

\begin{theorem} \label{th:fixedbisimdist}
Let $\M = (S, A, \tau, \rho, \ell)$ be an SMM.
\begin{enumerate}[topsep=0.5ex, noitemsep]
\item $\vartheta^\M$ is a pseudometric; 
\item for all $s,s' \in S$, $\vartheta^\M(s,s') = 0$ if and only if  $s \sim_\M s'$.
\end{enumerate}
\end{theorem}

\subsection{Complexity Results}
Provided that one has a way to compute the total variation distance between the residence-time probability distributions, Corollary~\ref{cor:modelstheta} allows one to apply the same idea of~\cite{BacciLM:tacas13} to obtained an on-the-fly algorithm for computing $\vartheta^\M$. 
This is the case for a relevant subclass of SMMs that includes MCs, CTMCs, and Markov models with uniformly or normally distributed residence time on states. For instance, the total variation distance between two exponential distributions with parameters $\lambda, \lambda' > 0$ is given by
\begin{equation*}
\tv{\mathit{Exp}(\lambda)}{\mathit{Exp}(\lambda')} = 
\begin{cases}
0 &\text{if $\lambda = \lambda'$} \\
\left| e^{\frac{\lambda' (\log{\lambda} - \log{\lambda'})}{\lambda' - \lambda}} - e^{\frac{\lambda (\log{\lambda} + \log{\lambda'})}{\lambda - \lambda'}} \right| & \text{otherwise}
\end{cases}
\end{equation*}

Moreover, by Lemma~\ref{lem:coupling-dist} we can also establish good theoretical results.

From a theoretical point of view, it is irrelevant whether the transition probabilities have rational values or not. However, for the complexity results that follow we assume that, for all $s \in S \setminus A$ and $s' \in S$, it holds $\tau(s)(s') \in \Q[] \cap [0,1]$.

\begin{lemma} \label{th:bisimpoly}
Let $\M = (S, A, \tau, \rho, \ell)$ be an SMM. Then, $\sim_\M$ can be computed in polynomial time in $\mathit{size}(\M)$, provided that, for all $s,s' \in S \setminus A$, the value $\tv{\rho(s)}{\rho(s')}$ can be computed in polynomial time in $\mathit{size}(\M)$.
\end{lemma}

The proof of Theorem~\ref{th:poly} is based on the same idea of~\cite{ChenBW12}. Specifically, $\vartheta^\M$ is characterized as the solution of a linear program that can be solved in polynomial time using the ellipsoid method~\cite{Schrijver86}.

\begin{theorem} \label{th:poly}
Let $\M = (S, A, \tau, \rho, \ell)$ be an SMM. Then, $\vartheta^\M$ can be computed in polynomial time in $\mathit{size}(\M)$, provided that, for all $s,s' \in S \setminus A$, the value $\tv{\rho(s)}{\rho(s')}$ can be computed in polynomial time in $\mathit{size}(\M)$.
\end{theorem}
\vspace{-3ex}
\section{Conclusions and Future Work} \label{sec:conclusions}
In this paper we took a step forward supporting approximate reasoning on Stochastic Markov Models in particular with respect to quantitative verification against linear real-time specifications expressed as MTL formulas or DTAs.

\begin{enumerate*}[label=\emph{(\roman*)}]
\item We showed that the real-time specifications expressed as MTL formulas or DTAs generate the same $\sigma$-algebra. Then using a suitable pseudometric between specifications, we establish the subclass of resetting single-clock DTAs to be dense in the whole $\sigma$-algebra.
\item We propose two behavioral distances between SMMs, defined as the maximal variation w.r.t.\ the probability of attaining respectively any MTL or DTA specification, and we showed that the two coincide to the total variation measured on the above mentioned $\sigma$-algebra.
\item We established the \textbf{NP}-hardness of computing these distances exactly or within a certain absolute error that depends on the size of the system.
\item Nevertheless, we defined a Kantorovich-based bisimilarity distance that over approximates these variation distances and can be computed in polynomial time. 
\end{enumerate*}

Each of the previous points are supported by practical motivations.

As future work we will investigate other possible logical characterization, e.g. considering CSL or $\MTL$ with continuous semantics. From the computational perspective, also motivated by our recent work~\cite{BacciLM:tacas13,BacciLM:mfcs13} on MCs and MDPs, we would like to implement an on-the-fly algorithm for computing $\vartheta^\M$ and develop a compositional theory for these models. 
To the best of our knowledge, whether the distance $\dist^\M$ is computable or not still remains an open problem.

\bibliographystyle{abbrv}
\bibliography{biblio}

\newpage

\appendix

\section{Technical proofs}

This section contains all the technical proofs that have been omitted in the paper and, in addition, some technical lemmas that are required in the technical development but are not the main exposition.

\begin{proof}[of Lemma~\ref{lem:MTLmeas}]
Immediate by Lemma~\ref{lem:sigmaEquiv} and $\sigma(\mathcal{T}_\M) \subseteq \Sigma_{\paths{S, 2^S}}$.
\qed
\end{proof}

\begin{proof}[of Lemma~\ref{lem:DTAmeas}]
Immediate by Lemma~\ref{lem:sigmaEquiv} and $\sigma(\mathcal{T}_\M) \subseteq \Sigma_{\paths{S, 2^S}}$.
\qed
\end{proof}

\begin{lemma} \label{lem:tailmeasurable}
Let $\M = (S,A,\tau,\rho,\ell)$ be an SMM and $tl_\M \colon \paths{S} \to \paths{S}$ be defined as $tl_\M(\pi) = \tail{\pi}{1}$. Then, $tl_\M$ is measurable w.r.t $(\paths{S}, \sigma(\mathcal{T}_\M))$.
\end{lemma}
\begin{proof}
It suffices to show that, for any $C \in \mathcal{T}_\M$, it holds $tl_\M^{-1}(C) \in \sigma(\mathcal{T}_\M)$. 
\begin{align*}
  tl_\M^{-1}(C) 
  &= \set{ \pi \in \paths{S} }{ tl_\M(\pi) \in C } \tag{def. pre-image} \\
  &=   
  \set{\pi \in \paths{S} }{\tail{\pi}{1} \in C } \tag{def. $tl_\M$} \,.
\intertext{by definition of $\mathcal{T}_\M$, $C = \cyl{[s_0]_{\eqlbl}, R_0, \dots, R_{n-1}, [s_n]_{\eqlbl}}$ for some $n \in \N$, and some $s_i \in S$ and $R_i \in \Sigma_{\R}$ ($i =0..n$)}
 &= \cyl{S, \R, [s_0]_{\eqlbl}, R_0, \dots, R_{n-1}, [s_n]_{\eqlbl}} \\
 &= \textstyle \bigcup_{s \in S} \cyl{[s]_{\eqlbl}, \R, [s_0]_{\eqlbl}, R_0, \dots, R_{n-1}, [s_n]_{\eqlbl}}
\end{align*}
hence $tl_\M^{-1}(C) \in \sigma(\mathcal{T}_\M)$ since it is a (finite) union of cylinders in $\mathcal{T}_\M$
\qed
\end{proof}

\begin{proof}[of Lemma~\ref{lem:sigmaEquiv}]
\begin{trivlist}
\item Proof of $\sigma(\denot{\MTL}) = \sigma(\mathcal{T}_\M)$: 

($\subseteq$) It suffices to prove that $\denot{\MTL} \subseteq \sigma(\mathcal{T}_\M)$.
We proceed by structural induction on the formulas $\varphi \in \MTL$ showing that $\denot{\varphi} \in \sigma(\mathcal{T}_\M)$.
\begin{itemize}[leftmargin=0.5ex, label={}, topsep=0.5ex,itemsep=0.5ex]
   \item \textbf{Atomic prop.} $\denot{a} = \set{ \pi }{ a \in \ell(\stateat{\pi}{0}) } = \bigcup \set{\cyl{[s]_{\eqlbl}}}{ s \in \ell^{-1}(\set{a}{})}$. Since $S$ is finite and $\cyl{[s]_{\eqlbl}} \in \mathcal{T}_\M$ for all $s \in S$, then $\denot{a}  \in \sigma(\mathcal{T}_\M)$. 
   \item \textbf{False.} $\denot{\false} = \emptyset \in \sigma(\mathcal{T}_\M)$.
   \item \textbf{Implication.} $\denot{\phi \to \psi} = \denot{ \neg \phi \vee \psi} = \denot{\phi}^c \cup \denot{\psi}$. By inductive hypothesis, $\denot{\phi}, \denot{\psi} \in \sigma(\mathcal{T}_\M)$, therefore 
   $\denot{\phi \to \psi} \in \Sigma_\M$.
   \item \textbf{Next.} Consider $\Next[{I}]{\phi}$.  The following hold
  \begin{align}
	\denot{\Next[{I}]{\phi}} &= 
	\set{\pi}{ \timeat{\pi}{0} \in I, \text{ and } \M, \tail{\pi}{1} \models \phi } 
	\tag{by def. of $\Next[]{}$} \\
	&= \set{\pi}{ \timeat{\pi}{0} \in I, \text{ and } \tail{\pi}{1} \in \denot{\phi} } 
	\tag{by def. of $\denot{\cdot}$} \\
	&= \set{\pi}{ \tail{\pi}{1} \in \denot{\phi} } \cap \cyl{S, I, S} \notag \\
	&= tl_\M^{-1}(\denot{\phi})  \cap \cyl{S , I, S}. \tag{by def. $tl_\M$}
  \end{align}
  $\cyl{S , I, S} = \bigcup_{s,s' \in S} \cyl{[s]_{\eqlbl}, I, [s']_{\eqlbl}} \in \sigma(\mathcal{T}_\M)$ since 
  $\cyl{[s]_{\eqlbl}, I, [s']_{\eqlbl}} \in \mathcal{T}_\M$ for all $s, s' \in S$ and $S$ is finite. By inductive hypothesis $\denot{\phi} \in \sigma(\mathcal{T}_\M)$, hence, by Lemma~\ref{lem:tailmeasurable}, $tl_\M^{-1}(\denot{\phi}) \in \sigma(\mathcal{T}_\M)$. Therefore $\denot{\Next[{I}]{\phi}} \in \sigma(\mathcal{T}_\M)$.
   \item \textbf{Until.} Consider $\denot{\Until[{[a,b]}]{\phi}{\psi}}$.  For $k > 0$ we define the set $\mathit{OnTime}@k$ as
\begin{equation*}
	\mathit{OnTime}@k =
	\bigcup \left\{ \cyl{X} \left|  
	\begin{aligned}
	&s_i \in S, \, t^{-}_i, t^{+}_i \in \Q, \text{ for } 0 \leq i \leq k, \\
	&\textstyle \sum_{i = 0}^{k-1} t^{-}_i \geq a, \, \sum_{i = 0}^{k-1} t^{+}_i \leq b, \, t^{-}_i \leq t^{+}_i, \\
	&X = [s_0]_{\eqlbl}, [t^{-}_0, t^{+}_0] ,\dots, [t^{-}_{k-1}, t^{+}_{k-1}],[s_n]_{\eqlbl}
	\end{aligned}
	\right. \right\}
\end{equation*}
Notice that $\mathit{OnTime}@k$ is a countable union of cylinders in $\mathcal{T}_\M$ (the number of unions is bounded by $|(S \times \Q^2)^{k+1}|$), hence it is a measurable set in $\sigma(\mathcal{T}_\M)$.

Now we show that 
\begin{equation}
\mathit{OnTime}@k = \{ \pi \mid  \forall i < k.\, \textstyle\sum_{i = 0}^{k-1} \timeat{\pi}{i} \in [a,b]\} \label{eq:ontime}
\end{equation} 
The inclusion from left to right trivially holds by definition of $\mathit{OnTime}@k$. As for the reverse inclusion, let $\pi$ be a timed path over $S$, such that $\timeat{\pi}{i} = t_i$ ($i=0..k-1$) and $\sum_{i = 0}^{k-1} t_i \in [a,b]$. We have to prove that there exist $t^{-}_i,t^{+}_i \in \Q$ such that 
$t^{-}_i \leq t_i \leq t^{+}_i$, $\sum_{i = 0}^{k-1} t^{-}_i \geq a$, and $\sum_{i = 0}^{k-1} t^{+}_i \leq b$.
When $k = 1$ it suffices to take $t^{-}_1 = a$ and $t^{+}_1 = b$. Assume $k > 1$.
Let $\Delta = 2h / 10^h$ for some $h \in \N$ large enough to satisfy the following two inequalities 
$\sum_{i = 0}^{k-1} t_i - \Delta > a$ and $\sum_{i = 0}^{k-2} t_i + \Delta < b$.
Let $t^{-}_i = t_i = t^{+}_i$ if $t_i \in \Q$, otherwise we choose some $t^{-}_i, t^{+}_i \in \Q$ that satisfy
\begin{align}
t^{-}_i < t_i < t^{+}_i,  && t^{-}_i > t_i - \Delta/2k \text{, and} && t^{+}_i < t_i + \Delta/2k \,.
\label{eq:intervalconstraints}
\end{align}
We proceed by showing that the constraints \eqref{eq:intervalconstraints} are sufficient to prove that $\sum_{i = 0}^{k-1} t^{-}_i \geq a$ and $\sum_{i = 0}^{k-1} t^{+}_i \leq b$, then we show how to pick such $t^{-}_i, t^{+}_i \in \Q$ in order to satisfy \eqref{eq:intervalconstraints}.
The following hold
\begin{align}
  \textstyle\sum_{i = 0}^{k-1} t_i - \Delta &< \textstyle\sum_{i = 0}^{k-1} (t^{-}_i + \Delta/2k) - \Delta =\tag{by \eqref{eq:intervalconstraints}} \\
  &= \textstyle\sum_{i = 0}^{k-1} t^{-}_i - \Delta/2 \leq \textstyle\sum_{i = 0}^{k-1} t^{-}_i \tag{by $\Delta \geq 0$}
\end{align}
By construction $\sum_{i = 0}^{k-1} t_i - \Delta > a$, hence $\sum_{i = 0}^{k-1} t^{-}_i > a$. Analogously, we have that 
\begin{align}
  \textstyle\sum_{i = 0}^{k-2} t_i + \Delta &> \textstyle\sum_{i = 0}^{k-1} (t^{+}_i - \Delta/2k) - \Delta =\tag{by \eqref{eq:intervalconstraints}} \\
  &= \textstyle\sum_{i = 0}^{k-2} t^{+}_i + (k+1)\Delta / 2k \geq \textstyle\sum_{i = 0}^{k-2} t^{+}_i \tag{by $\Delta \geq 0$}
\end{align}
By construction $\sum_{i = 0}^{k-2} t_i + \Delta < b$, hence $\sum_{i = 0}^{k-2} t^{+}_i < b$.

One can check that the constraints \eqref{eq:intervalconstraints} are easily satisfied if we pick 
\newcommand{\floor}[1]{ #1}
\begin{align*}
	 t^{-}_i = \lfloor t_i \rfloor + \frac{\lfloor 10^h \cdot \{ t_i \} \rfloor}{ 10^h } \,,
	 &&
	 t^{+}_i = \lfloor t_i \rfloor + \frac{\lfloor 10^h \cdot \{ t_i \} \rfloor + 1}{ 10^h } \,,
\end{align*}
where $\{ t_i \}$ denotes the fractional part of $t_i \not\in \Q$. 
This proves \eqref{eq:ontime}.
\begin{align}
	&\denot{\Until[{[a,b]}]{\phi}{\psi}} \notag \\
	&= \left\{ \pi \left|
	\begin{aligned}
		&\exists i > 0.\, \textstyle \sum_{k=0}^{i-1} \timeat{\pi}{i} \in [a,b], 
		\text{ and } \M, \tail{\pi}{i} \models \psi, \\
		 &\forall 0 \leq j < i. \, \M, \tail{\pi}{j} \models \phi
	\end{aligned} \right. \right\} \tag{by def. $\Until[]{}{}$} \\
	&= \left\{ \pi \left|
	\begin{aligned}
		&\exists i > 0.\, \textstyle \sum_{k=0}^{i-1} \timeat{\pi}{i} \in [a,b], 
		\text{ and } \tail{\pi}{i} \in \denot{\psi}, \\
		 &\forall 0 \leq j < i. \, \tail{\pi}{j} \in \denot{\phi}
	\end{aligned} \right. \right\} \tag{by def. $\denot{\cdot}$}\\
	&= \bigcup_{i > 0} \bigcap_{0 \leq j < i} ( (tl_\M^j)^{-1}(\denot{\phi}) \cap (tl_\M^i)^{-1}(\denot{\psi}) \cap \mathit{OnTime}@i ) \,. \tag{by def. $tl_\M$ and \eqref{eq:ontime}}
\end{align}
By inductive hypothesis $\denot{\phi}, \denot{\psi} \in \sigma(\mathcal{T}_\M)$, so that, by Lemma~\ref{lem:tailmeasurable} and since $\sigma(\mathcal{T}_\M)$ is closed under countable union and intersection, $\denot{\Until[{[a,b]}]{\phi}{\psi}} \in \sigma(\mathcal{T}_\M)$.
\end{itemize} 

\newcommand{\intcyl}{\cyl{S/_{\eqlbl},\mathcal{I}}}

($\supseteq$) Let $\mathcal{I}$ be the family of closed intervals in $\R$ with rational endpoints.
Clearly, $\intcyl \subseteq \mathcal{T}_\M$, hence $\sigma(\intcyl) \subseteq \sigma(\mathcal{T}_\M)$.
It is standard that $\sigma(\mathcal{I}) = \Sigma_{\R}$, and from it one can easily verify that $\mathcal{T}_\M \subseteq \sigma(\intcyl)$, therefore we have also $\sigma(\mathcal{T}_\M) \subseteq \sigma(\intcyl)$. This proves $\sigma(\intcyl) = \sigma(\mathcal{T}_\M)$. 

From this equality, to show the inclusion $\sigma(\mathcal{T}_\M) \subseteq\sigma(\denot{\MTL})$, it suffices to prove that $\intcyl \subseteq \sigma(\denot{\MTL})$.

Let define $Ap \colon AP \times \intcyl \to \MTL$ as
\begin{align*}
  Ap(a, \cyl{[s]_{\eqlbl}}) &= \left\{
    \begin{aligned}
    &a &&\text{if $a \in \ell(s)$} \\
    &\neg a &&\text{othewise}
    \end{aligned}\right.
  \\
  Ap(a, \cyl{[s]_{\eqlbl},I, X}) &= \left\{
    \begin{aligned}
    &a \wedge \Next[I]{Ap(a, \cyl{X})} &&\text{if $a \in \ell(s)$} \\
    &\neg a \wedge \Next[I]{Ap(a, \cyl{X})} &&\text{otherwise} \,,
    \end{aligned}\right.
\end{align*}

Let $C = \cyl{[s_0]_{\eqlbl}, I_0, \dots, I_{n-1}, [s_{n}]_{\eqlbl}} \in \intcyl$, one can prove by induction on $n$ that $\bigcap_{a \in AP} \denot{Ap(a,C)} = C$.
Since $\sigma(\denot{\MTL})$ is closed under countable intersection, we conclude that $C \in \sigma(\denot{\MTL})$.

\item Proof of $\sigma(\denot{\DTA}) = \sigma(\mathcal{T}_\M)$:

($\subseteq$) It suffices to show that $\denot{\DTA} \subseteq \sigma(\mathcal{T}_\M)$. This is proven in~\cite[Theorem~3.2]{KatoenLMCS11} and the proof can be left unchanged.

\renewcommand{\intcyl}{\cyl{S/_{\eqlbl},\mathcal{I}}}

($\supseteq$) Let $\mathcal{I}$ be the family of closed intervals in $\R$ with rational endpoints.
Clearly, $\intcyl \subseteq \mathcal{T}_\M$, hence $\sigma(\intcyl) \subseteq \sigma(\mathcal{T}_\M)$.
It is standard that $\sigma(\mathcal{I}) = \Sigma_{\R}$, and from it one can easily verify that $\mathcal{T}_\M \subseteq \sigma(\intcyl)$, therefore we have also $\Sigma_\M \subseteq \sigma(\intcyl)$. This proves $\sigma(\intcyl) = \sigma(\mathcal{T}_\M)$. 

Hence, for $\Sigma_\M \subseteq \sigma(\denot{\DTA})$ it suffices to prove $\intcyl \subseteq \sigma(\denot{\DTA})$. Let $C = \cyl{[s_0]_{\eqlbl}, I_0, \dots, I_{n-1}, [s_{n}]_{\eqlbl}} \in \intcyl$. We want to define a DTA $\A = (Q, 2^{AP}, q_0, F, \to)$ such that $\denot{\A} = C$. This can be obtained by setting 
$Q = \set{q_0, \dots, q_n}{}$, $F = \set{q_n}{}$, and, using one shared clock $x \in \clk$ in each guard,
\begin{align*}
  {\to} = 
  &\set{ (q_i, \ell(s_i), g_i, \clk, q_{i+1}) }{ g_i = a \leq x \leq b \text{ for } I_i = [a,b], 0 \leq i \leq n} \\
  &{} \cup \set{ (q_n, a, \emptyset, \clk, q_n) }{ a \subseteq AP } \,.
\end{align*}
It is easy to see that the only accepted timed paths $\pi \in \L{\A}$ are such that 
$\prefix{\pi}{n} = \ell(s_0), t_0, \dots, t_{n-1}, \ell(s_n)$, and $t_i \in I_i$ ($0 \leq i \leq n-1$), because clocks are always resetting. So the thesis.
\qed
\end{trivlist}
\end{proof}

\begin{proof}[of Lemma~\ref{lem:densefield}]
We have to show that $\cl{\mathcal{F}} = \Sigma$. The closure of $\mathcal{F}$ under $d_\mu$ is given by
\begin{equation*}
  \cl{\mathcal{F}} = 
  \set{E \in \Sigma}{ \forall \varepsilon > 0. \; \exists F \in \mathcal{F}. \; d_\mu(E,F) < \varepsilon} \,.
\end{equation*}
Clearly $\cl{\mathcal{F}} \subseteq \Sigma$. The converse inclusion follows by $\mathcal{F} \subseteq \cl{\mathcal{F}}$ and $\Sigma = \sigma(\mathcal{F})$, showing that $\cl{\mathcal{F}}$ is a $\sigma$-algebra, i.e., closed under complement and countable union.
\begin{itemize}[leftmargin=0.5ex, label={}, topsep=0.5ex,itemsep=0.5ex]
  \item \textit{Complement.} Let $E \in \cl{\mathcal{F}}$. We want to show that $E^c \in \cl{\mathcal{F}}$, where $E^c := \Sigma \setminus E$ denotes the complement of $E$ in $\Sigma$.   
Let $\varepsilon > 0$. By $E \in \cl{\mathcal{F}}$, there exists $F \in \mathcal{F}$ such that $d_\mu(E, F) < \varepsilon$. Moreover, note that $E \symdiff F = E^c \symdiff F^c$, so 
\begin{equation*}
  d_\mu(E^c, F^c) = \mu(E^c \symdiff F^c) =  \mu(E \symdiff F) = d_\mu(E,F) \,,
\end{equation*}
and $d_\mu(E^c, F^c) < \varepsilon$. By hypothesis, $\mathcal{F}$ is a field, hence $F^c \in \mathcal{F}$. Due to the generality of $\varepsilon > 0$, this proves $E^c \in \cl{\mathcal{F}}$. 
  
  \item \textit{Countable Union.} Let $\set{E_i}{i \in \N} \subseteq \cl{\mathcal{F}}$. We want to show that $\bigcup_{i \in \N} E_i \in \cl{\mathcal{F}}$. Let $\varepsilon > 0$. To prove the thesis it suffices to show that  following statements hold:
  \begin{enumerate}[label={\alph*)}, noitemsep, topsep=0.5ex, fullwidth]
    \item there exists $k \in \N$, such that 
    $d_\mu(\bigcup_{i \in \N} E_i, \bigcup_{i = 0}^k E_i) < \frac{\varepsilon}{2}$;
    \item for all $n \in \N$, there exist $F_0,\dots, F_n \in \mathcal{F}$, such that $d_\mu(\bigcup_{i = 0}^n E_i, \bigcup_{i = 0}^n F_i) < \frac{\varepsilon}{2}$.
  \end{enumerate}
  Indeed, by applying the triangular inequality on (a) and (b), we have that there exist $k \in \N$ and $F_0,\dots, F_k \in \mathcal{F}$ such that
  \begin{equation*}
    \textstyle
    d_\mu(\bigcup_{i \in \N} E_i, \bigcup_{i = 0}^k F_i) 
    \leq d_\mu(\bigcup_{i \in \N} E_i, \bigcup_{i = 0}^k E_i) + d_\mu(\bigcup_{i = 0}^k E_i, \bigcup_{i = 0}^k F_i)
    < \varepsilon \,.
  \end{equation*}
  Since, by hypothesis, $\mathcal{F}$ is a field, we also have that $\bigcup_{i = 0}^k F_i \in \mathcal{F}$. Therefore, due to the generality of $\varepsilon > 0$, we will obtain that $\bigcup_{i \in \N} E_i \in \cl{\mathcal{F}}$. 
  
  (a). Since $(\bigcup_{i = 0}^n E_i)_{n \in \N}$ is a countable increasing sequence in $\Sigma$ converging to $\bigcup_{i \in \N} E_i$, by $\omega$-continuity from below of $\mu$, we have that $(\mu(\bigcup_{i = 0}^n))_{n \in \N}$ converges in $\R[]$ to $\mu(\bigcup_{i \in \N} E_i)$. This means that there exists and index $k \in \N$ such that 
  \begin{equation*}
    \textstyle
    |\mu(\bigcup_{i \in \N} E_i) - \mu(\bigcup_{i = 0}^k E_i)| < \frac{\varepsilon}{2}
  \end{equation*}
  By $\bigcup_{i \in \N} E_i \subseteq \bigcup_{i = 0}^k E_i$ and monotonicity, additivity, and finiteness of $\mu$,
  \begin{align*}
    \textstyle
    d_\mu(\bigcup_{i \in \N} E_i, \bigcup_{i = 0}^k E_i)
      &= \textstyle  
      \mu(\bigcup_{i \in \N} E_i \symdiff \bigcup_{i = 0}^k E_i) \\
      &= \textstyle  
      \mu(\bigcup_{i \in \N} E_i \setminus \bigcup_{i = 0}^k E_i) \\
      &= \textstyle
      \mu(\bigcup_{i \in \N} E_i) - \mu(\bigcup_{i = 0}^k E_i) < \textstyle \frac{\varepsilon}{2} \,. 
  \end{align*}
  
  (b). Let $n \in \N$. By $E_0, \dots, E_n \in \cl{\mathcal{F}}$, there exists $F_0, \dots, F_n \in \mathcal{F}$ such that $d_\mu(E_i, F_i) < \frac{\varepsilon}{2n}$. Moreover, note that $\bigcup_{i = 0}^n E_i \symdiff \bigcup_{i = 0}^n F_i \subseteq \bigcup_{i = 0}^n (E_i \symdiff F_i)$, so that by monotonicity and sub-additivity of $\mu$ we have
  \begin{align*}
    \textstyle
    d_\mu(\bigcup_{i = 0}^n E_i, \bigcup_{i = 0}^k F_i)
      &= \textstyle  
      \mu(\bigcup_{i = 0}^n E_i \symdiff \bigcup_{i = 0}^n F_i) \\
      &\leq \textstyle  
      \mu(\bigcup_{i = 0}^n (E_i \symdiff F_i) ) \\
      &\leq \textstyle  
      \sum_{i=0}^n \mu(E_i \symdiff F_i) 
      < \textstyle  
      \sum_{i=0}^n \frac{\varepsilon}{2n} = \frac{\varepsilon}{2} \,.
  \end{align*}
\end{itemize}
\qed
\end{proof}

\begin{proof}[of Lemma~\ref{lem:sDTARfield}]
One inclusion follows obviously since $\denot{\sDTAR} \subseteq \denot{\DTA}$. The converse inclusion is already proven in Lemma~\ref{lem:sigmaEquiv} which uses single-clock resetting DTAs only. 
\qed
\end{proof}

\begin{proof}[of Theorem~\ref{th:1RDTAapproximant}]
By Lemmas~\ref{lem:sigmaEquiv} and \ref{lem:sDTARfield}, $\denot{\sDTAR}$ generates $\Sigma_\M$, moreover it is a field. Therefore, by Lemma~\ref{lem:densefield},  $\denot{\sDTAR}$ is dense in $(\Sigma_\M, d_\mu)$, for all finite measures $\mu$ over $(\paths{S},\Sigma_\M)$.
To prove the thesis it suffices to show that the function $\mu \colon \Sigma \to \R[]$ is continuous. Let $E$ and $F$ be arbitrary measurable sets in $\Sigma$, then
\begin{align}
  \mu(E) 
    &= \mu(E \setminus F) + \mu(E \cap F)  \tag{$\mu$ additive} \\
    &\leq \mu((E \setminus F) \cup (F \setminus E)) + \mu(F)  \tag{$\mu$ monotone} \\
    &= \mu(E \symdiff F) + \mu(E)  \tag{by def} \\
    &= \mu(E \symdiff F) +  \nu(E \symdiff F) + \mu(F)  \tag{$\nu$ positive} \\
    &= d_\mu(E, F) + \mu(F) \,.  \tag{by def}
\end{align}
This implies that, for all $E, F \in \Sigma$, $d(E,F) \geq |\mu(E) - \mu(F)|$, hence $\mu \colon \Sigma \to \R[]$ is $1$-Lipschitz continuous.
\qed
\end{proof}

\begin{proposition} \label{prop:supofclosures}
Let $A \subseteq \R[]$ be a bounded nonempty set. Then, 
\begin{enumerate}[label={(\roman*)}, leftmargin=*, topsep=0.5ex, itemsep=0.5ex]
  \item \label{prop:supofclosures1} $\sup A \in \cl{A}$;
  \item \label{prop:supofclosures2} $\sup A = \sup \cl{A}$.
\end{enumerate}
\end{proposition}
\begin{proof}
First, notice that since $A \neq \emptyset$ and is bounded, by Dedekind axiom, the supremum of $A$ (and $\cl{A}$) in $\R[]$ exists. Moreover, recall that, for any $B \subseteq \R[]$,
\begin{equation*}
\cl{B} = ad(B) := 
  \set{x \in \R[]}{\forall \varepsilon > 0. \; (x - \varepsilon, x + \varepsilon) \cap B \neq \emptyset} \,,
\end{equation*}
where $ad(B)$ denotes the set of points \emph{adherent} to $B$.

Let $\alpha = \sup A$. \ref{prop:supofclosures1} We prove that $\alpha \in \cl{A}$. Let $\varepsilon > 0$, then $\alpha - \varepsilon$ is not an upper bound for $A$. This means that there exists $x \in A$ such that $\alpha - \varepsilon < x \leq \alpha$ and, in particular, that $x \in (\alpha-\varepsilon, \alpha+\varepsilon) \cap A$. Therefore $\alpha \in \cl{A}$. \ref{prop:supofclosures2} Let $\beta = \sup \cl{A}$. By $A \subseteq \cl{A} = \cl{\cl{A}}$ and \ref{prop:supofclosures1}, we have $\alpha \leq \beta \in \cl{A}$. We prove that $\alpha = \beta$. Assume by contradiction that $\alpha \neq \beta$ and let $\varepsilon := \beta - \alpha$. Clearly $\varepsilon > 0$, so that, by $\beta \in \cl{A}$, we have that $(\beta - \varepsilon, \beta+\varepsilon) \cap A \neq \emptyset$. This means that there exists $x \in A$ such that $\alpha = \beta - \varepsilon < x$, in contradiction with the hypothesis that $\alpha = \sup A$. 
\qed
\end{proof}

\begin{proposition} \label{prop:closureimage} 
Let $f \colon X \to Y$ be continuous and $A \subseteq X$, then $\cl{f(A)} = \cl{f(\cl{A})}$.
\end{proposition}
\begin{proof}
($\supseteq$) A function $f \colon X \to Y$ is continuous iff for all $B \subseteq X$, $f(\cl{B}) \subseteq \cl{f(B)}$. Therefore $f(\cl{A}) \subseteq \cl{f(A)}$. Since $\cl{f(A)}$ is closed, we have $\cl{f(\cl{A})} \subseteq \cl{f(A)}$. ($\subseteq$) The result follows by $A \subseteq \cl{A}$ and monotonicity of $f(\cdot)$ and $\cl{(\cdot)}$. \qed
\end{proof}

\begin{proposition} \label{supremumondense} 
Let $X$ be nonempty, $f \colon X \to \R[]$ be a bounded continuous real-valued function, and $D \subseteq X$ be dense in $X$. Then $\sup f(D) = \sup f(X)$.
\end{proposition}
\begin{proof}
Notice that, since $X \neq \emptyset$ and $f$ is bounded, by Dedekind axiom, both $\sup f(D)$ and $\sup f(X)$ exist. By Propositions~\ref{prop:supofclosures}, \ref{prop:closureimage}, and $\cl{D} = X$, we have
\begin{equation*}
\sup f(D) 
  \stackrel{\text{(Prop.\ref{prop:supofclosures})}}{=} \sup \cl{f(D)} 
  \stackrel{\text{(Prop.\ref{prop:closureimage})}}{=} \sup \cl{f(\cl{D})} 
  = \sup \cl{f(X)} 
  \stackrel{\text{(Prop.\ref{prop:supofclosures})}}{=} \sup f(X) \,,
\end{equation*}
which proves the thesis. \qed
\end{proof}

\begin{proof}[of Lemma~\ref{lem:vardistonfield}]
Consider the pseudometric $d(E,F) = \mu(E \symdiff F) + \nu(E \symdiff F)$ on $\Sigma$, i.e., the Fr\'echet-Nikodym pseudometric w.r.t. $\mu + \nu$. By Lemma~\ref{lem:densefield}, we know that $\mathcal{F}$ is dense in $(\Sigma, d)$. For any nonempty set $Y$ and any bounded continuous real-valued function $f \colon Y \to \R[]$, if $D \subseteq Y$ is dense then $\sup f(D) = \sup f(X)$ (see Proposition~\ref{supremumondense}). Therefore, to prove the thesis it suffices to show that the function $|\mu - \nu| \colon \Sigma \to \R[]$ is bounded and continuous. Let $E$ and $F$ be arbitrary measurable sets in $\Sigma$, then
\begin{align}
  \mu(E) 
    &= \mu(E \setminus F) + \mu(E \cap F)  \tag{$\mu$ additive} \\
    &\leq \mu((E \setminus F) \cup (F \setminus E)) + \mu(F)  \tag{$\mu$ monotone} \\
    &= \mu(E \symdiff F) + \mu(E)  \tag{by def} \\
    &= \mu(E \symdiff F) +  \nu(E \symdiff F) + \mu(F)  \tag{$\nu$ positive} \\
    &= d(E, F) + \mu(F) \,.  \tag{by def}
\end{align}
This implies that, for all $E, F \in \Sigma$, $d(E,F) \geq |\mu(E) - \mu(F)|$, hence $\mu \colon \Sigma \to \R[]$ is $1$-Lipschitz continuous. Analogously, also $\nu \colon \Sigma \to \R[]$ is $1$-Lipschitz continuous. Then, continuity of $|\mu - \nu| \colon \Sigma \to \R[]$ follows by composition of continuous functions. Moreover, $|\mu - \nu|$ is bounded since, by hypothesis, $\mu$ and $\nu$ are finite.
\qed
\end{proof}

\begin{proof}[of Theorem~\ref{thm:equiDist}]
The thesis follows by Lemmas~\ref{lem:vardistonfield} and \ref{lem:sigmaEquiv}, noticing that $\denot{\MTL}$ and $\denot{\DTA}$ are fields. \qed
\end{proof}

\begin{proof}[of Theorem~\ref{th:tracedist}]
($\supseteq$) Immediate by $\mathcal{T}_\M \subseteq \Sigma_\M$. ($\subseteq$) It follows by Hahn-Kolmogorov extension theorem, by noticing that the family $\mathcal{F}$ consisting of all finite unions of trace cylinders in $\cyl{S/_{\equiv_\ell}, \Sigma_{\R}}$ is indeed a field, and since $S/_{\equiv_\ell}$ has only pairwise disjoint subsets, the uniqueness of the extension w.r.t.\ $\mathcal{F}$ implies the uniqueness of the extension w.r.t. $\cyl{S/_{\equiv_\ell}, \Sigma_{\R}}$.
\qed
\end{proof}

\begin{proof}[of Lemma~\ref{lem:LpDist}]
The thesis follows from the following equalities:
\begin{align}
&\textstyle \sup_{E \in \Sigma_\M} | \Pr{s}(E) - \Pr{s'}(E) | = 
\sup_{E \in \sigma(\mathcal{W}_\mathcal{M})} | \Pr{s}(E) - \Pr{s'}(E) | \,, \label{eq:words}
\\[1ex]
& \textstyle 2 \cdot \sup_{E \in \sigma(\mathcal{W}_\mathcal{M})} | \Pr{s}(E) - \Pr{s'}(E) | = 
\sum_{ E \in \mathcal{W}_\mathcal{M} } {|\Pr{s}(E) - \Pr{s'}(E)|} \, . \label{eq:LpDist}
\end{align}

\begin{description}
\item[Equation~\eqref{eq:words}:] ($\geq$) directly follows by $\sigma(\mathcal{W}_\M) \subseteq \sigma(\mathcal{T}_\M) = \Sigma_\M$. 

($\leq$) We prove by induction on $n \in \N$, that for any pair of cylinder sets $C_n,C'_n$
of the form $\cyl{s_0, R_0, \dots, R_{n-1},s_n}$ and $\cyl{s_0, \R, \dots, \R,s_n}$ respectively, 
it holds that $|\Pr{s}(C_n) - \Pr{s'}(C_n)| \leq |\Pr{s}(C'_n) - \Pr{s'}(C'_n)|$. The base case, $n=0$, holds trivially.
For the inductive step, assume that $n \geq 1$ and, w.l.g.\, that $\Pr{s}(C_n) \geq \Pr{s'}(C_n)$. Then, the following hold 
\begin{align*}
&\Pr{s}(C_n) - \Pr{s'}(C_n) \tag{by def of $\Pr{}$} \\
&= P(s_{n-1},R_{n-1},s_{n}) \cdot \big( \Pr{s}(C_{n-1}) - \Pr{s'}(C_{n-1}) \big) \tag{by monotonicity} \\
&\geq P(s_{n-1},\R,s_{n}) \cdot \big( \Pr{s}(C_{n-1}) - \Pr{s'}(C_{n-1}) \big) \tag{by ind.\ hp.} \\
&\geq P(s_{n-1},\R,s_{n}) \cdot \big( \Pr{s}(C'_{n-1}) - \Pr{s'}(C'_{n-1}) \big) \tag{by def of $\Pr{}$} \\
& = \Pr{s}(C'_n) - \Pr{s'}(C'_n) 
\end{align*}
By Hahn-Kolmogorow extension theorem and Hahn decomposition theorem on signed measures it follows that 
for all $E \in \Sigma_\M$ there exists $E' \in \sigma(\mathcal{W}_\M)$ such that $| \Pr{s}(E) - \Pr{s'}(E) | \leq 
| \Pr{s}(E') - \Pr{s'}(E') |$. This proves \eqref{eq:words}.

\item[Equation~\eqref{eq:LpDist}:] 
Since $S$ is finite, $\mathcal{W}_\M$ has countably many elements, moreover, since $\eqlbl$ is an equivalence relation, they are also pairwise disjoint. Therefore, every measurable set $E \in \sigma(\mathcal{W}_\M)$ can be expressed as a countable union of cylinders taken from $\mathcal{W}_\M$. 
Let $\mathcal{P}$ be the family of all cylinders $C \in \mathcal{W}_\M$ such that $\Pr{s}(C) \geq \Pr{s'}(C)$. By Hahn decomposition theorem we have that $P = \bigcup \mathcal{P}$ is a positive set for the signed measure $\Pr{s} - \Pr{s'}$, that is
\begin{equation}
\textstyle \sup_{E \in \sigma(\mathcal{W}_\mathcal{M})} | \Pr{s}(E) - \Pr{s'}(E) | 
= \Pr{s}(P) - \Pr{s'}(P) \, . \label{eq:positiveset}
\end{equation}
Now we are ready to prove \eqref{eq:LpDist}
\begin{align*}
& \textstyle 2 \cdot \sup_{E \in \sigma(\mathcal{W}_\mathcal{M})} | \Pr{s}(E) - \Pr{s'}(E) | =  
	\tag{by \eqref{eq:positiveset}} \\
& = 2 \cdot \big( \Pr{s}(P) - \Pr{s'}(P) \big) = \\
& = \Pr{s}(P) - \Pr{s'}(P) + \Pr{s'}(P^c) - \Pr{s}(P^c) 
	\tag{$\sigma$-additivity} \\
& = \textstyle \sum_{C \in \mathcal{P}} \big( \Pr{s}(C) - \Pr{s'}(C) \big)
+ \sum_{C \not\in \mathcal{P}} \big( \Pr{s'}(C) - \Pr{s}(C) \big) \\
& = \textstyle \sum_{C \in \mathcal{W}_\M} | \Pr{s}(C) - \Pr{s'}(C) | \, .
\end{align*}
\end{description} 
\qed
\end{proof}

\begin{proof}[of Theorem~\ref{th:NPhardness}]
Let $\mathcal{G} = (V,E)$ be an undirected graph with $n$ vertices, i.e., $n = |V|$, and let fix $\kappa \in \Delta(\R)$. 
For simplicity, we assume $V = \{1,\dots, n\}$ and we use increasing chains of vertices $v_1 < \dots < v_n$ to represent their set.

Our construction will make use of some gadgets. Two gadgets can be composed to form a new one, this is done by ``gluing'' together two link points, namely, a sink and point is connected to a source link point forming transitions as depicted below
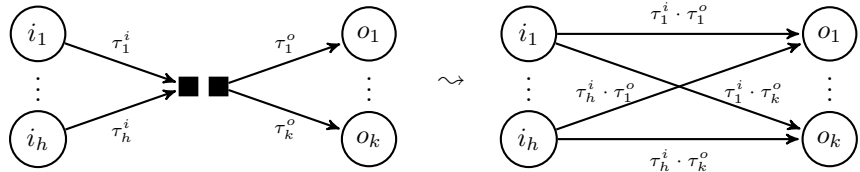
\begin{align*}
&\begin{tikzpicture}[->,>=stealth',shorten >=1pt,auto,node distance=1.2cm,
  thick,vertex node/.style={circle,draw},link node/.style={draw,fill=black},baseline]
   \node[link node] (linkend) {}; 
   \node[vertex node] (i1) at ($(linkend)+(-2,0.7)$) {$i_1$}; 
   \node[vertex node] (ih) at ($(linkend)-(2,0.7)$) {$i_h$};
   \node at ($(i1.south)!0.4!(ih.north)$) {\vdots};
   \node[link node] (linkstart) at ($(linkend)+(0.4,0)$) {}; 
   \node[vertex node] (o1) at ($(linkstart)+(2,0.7)$)  {$o_1$}; 
   \node[vertex node] (ok) at ($(linkstart)+(2,-0.7)$) {$o_k$};
   \node at ($(o1.south)!0.4!(ok.north)$) {\vdots};
  \path[every node/.style={font=\sffamily\scriptsize}]
    (i1) edge[above] node {$\tau^{i}_{1}$} (linkend) 
    (ih) edge[below] node {$\tau^{i}_{h}$} (linkend) 
    (linkstart) edge[above] node {$\tau^{o}_{1}$} (o1) 
    (linkstart) edge[below] node {$\tau^{o}_{k}$} (ok);
\end{tikzpicture}
&&
\leadsto
&&
\begin{tikzpicture}[->,>=stealth',shorten >=1pt,auto,node distance=1.2cm,
  thick,vertex node/.style={circle,draw},link node/.style={draw,fill=black},baseline]
   \node (link) {}; 
   \node[vertex node] (i1) at ($(link)+(-2,0.7)$) {$i_1$}; 
   \node[vertex node] (ih) at ($(link)-(2,0.7)$) {$i_h$};
   \node at ($(i1.south)!0.4!(ih.north)$) {\vdots};
   \node[vertex node] (o1) at ($(link)+(2,0.7)$)  {$o_1$}; 
   \node[vertex node] (ok) at ($(link)+(2,-0.7)$) {$o_k$};
   \node at ($(o1.south)!0.4!(ok.north)$) {\vdots};
  \path[every node/.style={font=\sffamily\scriptsize}]
    (i1) edge[above] node {$\tau^{i}_{1} \cdot \tau^{o}_{1}$} (o1) 
    (i1) edge[above, pos=0.8] node {$\tau^{i}_{1} \cdot \tau^{o}_{k}$} (ok) 
    (ih) edge[above, pos=0.2] node {$\tau^{i}_{h} \cdot \tau^{o}_{1}$} (o1) 
    (ih) edge[below] node {$\tau^{i}_{h} \cdot \tau^{o}_{k}$} (ok);
\end{tikzpicture}
\end{align*}

Moreover, a gadget can be rescaled by a factor $\varepsilon \in [0,1]$, denoted as $M[\varepsilon]$ as  
\begin{center}
\begin{tikzpicture}[->,>=stealth',shorten >=1pt,auto,node distance=1.2cm,
  thick,vertex node/.style={circle,draw},link node/.style={draw,fill=black}, baseline]
  \node[link node] (start) {};
  \node[link node] (l1) at ($(start)+(2,0.5)$) {};
  \node (l2)at ($(start)+(2,-0.5)$) {};
  \node[link node] (end) at ($(start)+(6,0)$) {};
  \node[link node] (l3) at ($(end)+(-2,0.5)$) {};
  \node (l4)at ($(end)+(-2,-0.5)$) {};
  
  \node (Lv1) at ($(l1)!0.5!(l3)$) {$M$};
  \draw[dashed] ($(l1.center)+(0,.4)$) rectangle ($(l3.center)-(0,.4)$);
  \node[vertex node] (Lvn) at ($(l2)!0.5!(l4)$) {$\beta$};
  
  \path[every node/.style={font=\sffamily\scriptsize}]
    (start) edge node[above] {$\varepsilon$} (l1)
    (start) edge node[below] {$1-\varepsilon$} (Lvn)
    (l3) edge node[above] {1} (end) 
    (Lvn) edge node[below] {1} (end);
\end{tikzpicture}
\end{center}
where $\beta$ is a label that doesn't occur in $M$.

Now we are ready to introduce the construction of the SMMs that will be used in the proof. 
We start by describing the SMM $\M_\mathcal{G}$. In $\M_\mathcal{G}$ each state has residence time distribution $\kappa$ and is labelled over the alphabet $V \cup \{\alpha,\omega\}$. Roughly, the purpose of $\M_\mathcal{G}$ is to allow the identification of each clique in $\mathcal{G}$ by measuring certain word cylinders in $\M_\mathcal{G}$.
The model $\M_\mathcal{G}$ consists of a start-state and an (absorbing) end-state, respectively labelled with $\alpha$ and $\omega$, and $n$ different gadgets $L(v)$ associated with each vertex $v \in V$. Each gadget $L(v)$ is connected with the start-state with entering probability $2^{\mathit{deg}(v)}/\gamma$, where $\gamma = \sum_{v \in V} 2^{\mathit{deg}(v)}$, and to the end-state with exit probability $1$, as depicted below
\begin{center}
\begin{tikzpicture}[->,>=stealth',shorten >=1pt,auto,node distance=1.2cm,
  thick,vertex node/.style={circle,draw},link node/.style={draw,fill=black}, baseline]
  \node[vertex node] (start) {$\alpha$};
  \node[link node] (l1) at ($(start)+(2,0.7)$) {};
  \node[link node] (l2)at ($(start)+(2,-0.7)$) {};
  \node[vertex node] (end) at ($(start)+(6,0)$) {$\omega$};
  \node[link node] (l3) at ($(end)+(-2,0.7)$) {};
  \node[link node] (l4)at ($(end)+(-2,-0.7)$) {};
  
  \node (Lv1) at ($(l1)!0.5!(l3)$) {$L(v_1)$};
  \draw[dashed] ($(l1.center)+(0,.4)$) rectangle ($(l3.center)-(0,.4)$);
  \node (Lvn) at ($(l2)!0.5!(l4)$) {$L(v_n)$};
  \draw[dashed] ($(l2.center)+(0,.4)$) rectangle ($(l4.center)-(0,.4)$);
  \node at ($(Lv1)!0.4!(Lvn)$) {$\vdots$};
  
  \path[every node/.style={font=\sffamily\scriptsize}]
    (start) edge node[above] {$\frac{2^{\mathit{deg}(v_1)}}{\gamma}$} (l1)
    (start) edge node[below] {$\frac{2^{\mathit{deg}(v_n)}}{\gamma}$} (l2)
    (l3) edge node[above] {1} (end) 
    (l4) edge node[below] {1} (end);
\end{tikzpicture}
\end{center}

The purpose of a gadget $L(v)$ is to measure with uniform probability all the cylinders of the form 
$\cyl{s_1,\R,\dots,\R,s_k}$ for some $1 \leq k \leq n$, such that
\begin{itemize}
\item $\ell(s_1) < \dots < \ell(s_k)$ is an increasing sequence of vertices;
\item for all $1 \leq i \leq k$, either $\ell(s_i) = v$ or $(v, \ell(s_i)) \in E$;
\item $\ell(s_i) = v$ for some $1 \leq i \leq k$.
\end{itemize}
$L(v)$ is the sequential composition 
of gadgets $H_v(1)\dots H_v(n)$ where
\begin{equation*}
H_v(u) =
\begin{cases}
\begin{tikzpicture}[->,>=stealth',shorten >=1pt,auto,node distance=1.2cm,
  thick,vertex node/.style={circle,draw},link node/.style={draw,fill=black}, baseline]
  \node[link node] (l1) {};
  \node[vertex node] (v1) [right of=l1] {$v$};
  \node[link node] (l2) [right of=v1] {};
  \path[every node/.style={font=\sffamily\scriptsize}]
    (l1) edge node {1} (v1) 
    (v1) edge node {1} (l2);
\end{tikzpicture} & \text{if $u = v$}
\\[2ex]
\begin{tikzpicture}[->,>=stealth',shorten >=1pt,auto,node distance=1.2cm,
  thick,vertex node/.style={circle,draw},link node/.style={draw,fill=black},baseline]
  \node[link node] (l2) [right of=v1] {};
  \node[vertex node] (v2) [right of=l2] {$u$};
  \node[link node] (l3) [right of=v2] {};
  \path[every node/.style={font=\sffamily\scriptsize}]
    (l2) edge node {$\frac{1}{2}$} (v2) 
    (l2) edge[bend right] node [below] {$\frac{1}{2}$} (l3)
    (v2) edge node {1} (l3);
\end{tikzpicture} & \text{if $(v, u) \in E$}
\\[2ex]
\begin{tikzpicture}[->,>=stealth',shorten >=1pt,auto,node distance=1.2cm,
  thick,vertex node/.style={circle,draw},link node/.style={draw,fill=black}, baseline]
  \node[link node] (l1) {};
  \node (v1) [right of=l1] {};
  \node[link node] (l2) [right of=v1] {};
  \path[every node/.style={font=\sffamily\scriptsize}]
    (l1) edge node {1} (l2);
\end{tikzpicture} & \text{otherwise}
\end{cases}
\end{equation*}
Intuitively, the $L$-gadgets are used to count the number of vertices in a maximum clique of $\mathcal{G}$. Indeed, for any $v \in V$, it holds that
\begin{itemize}
\item all the paths in $L(v)$ have the same probability and, 
\item for each increasing sequence $V'$ of vertices there is at most one path in $L(v)$ that generates it. In particular, those which are recognized always have $v$ and cannot contain any vertex which is not adjacent to $v$ in $\mathcal{G}$.
\end{itemize}

Note that the number of paths in any gadget $L(v)$ is $2^{\mathit{deg}(v)}$. Therefore, $\M_\mathcal{G}$ measures with uniform probability, $1/\gamma$, each path from the start-state, say $s$, to the end-state. Therefore
\begin{equation}
\Pr[\M_\mathcal{G}]{s}(\alpha v_1 \cdots v_k \omega) = k / \gamma
\iff \{v_1,\dots,v_k\} \text{ is a clique in $\mathcal{G}$} \label{eq:clique}
\end{equation}
where, by abuse of notation, we denote with the string $\sigma_1 \cdots \sigma_n$ the word cylinder set $\cyl{ \ell^{-1}(\sigma_1),\R, \dots,\R,\ell^{-1}(\sigma_n) }$. 

Let now describe the construction of another SMM, denoted by $\M_V$. As before, each state in $\M_V$ has residence time distribution $\kappa$ and is labelled over the alphabet $V \cup \{\alpha,\omega\}$. The purpose of $\M_V$ is to generate with uniform probability, $1/ 2^{n}$, all the increasing sequences of vertices in $V$. This is achieved by constructing $\M_V$ as the following sequential composition of gadgets
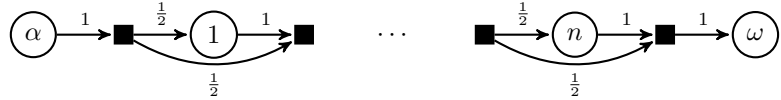
\begin{center}
\begin{tikzpicture}[->,>=stealth',shorten >=1pt,auto,node distance=1.2cm,
  thick,vertex node/.style={circle,draw},link node/.style={draw,fill=black}]
  \node[vertex node] (start) [left of=l1] {$\alpha$};
  \node[link node] (l1) {};
  \node[vertex node] (v1) [right of=l1] {$1$};
  \node[link node] (l2) [right of=v1] {};
  \node(v2) [right of=l2] {$\cdots$};
  \node[link node] (l3) [right of=v2] {};
  \node[vertex node] (v3) [right of=l3] {$n$};
  \node[link node] (l4) [right of=v3] {};
  \node[vertex node] (end) [right of=l4] {$\omega$};
  %
%
  \path[every node/.style={font=\sffamily\scriptsize}]
    (start) edge node {$1$} (l1)
    (l1) edge node {$\frac{1}{2}$} (v1) 
    (l1) edge[bend right] node [below] {$\frac{1}{2}$} (l2)
    (v1) edge node {$1$} (l2)
    (l3) edge node {$\frac{1}{2}$} (v3) 
    (l3) edge[bend right] node [below] {$\frac{1}{2}$} (l4)
    (v3) edge node {$1$} (l4)
    (l4) edge node {$1$} (end);
\end{tikzpicture}
\end{center}

Now we are ready to show that if one can compute $\dist$ in polynomial time, than he can solve Max Clique in polynomial time too. 
Let $x_i$ be the number of word cylinders $w \in \alpha V^* \omega$ such that $\Pr[\M_\mathcal{G}]{s}(w) = i/\gamma$. If we know the maximum $j$ such that $x_j \neq 0$, by \eqref{eq:clique} we have that the maximum clique size of $\mathcal{G}$ is $j$.
Consider the vector $\vec{x}$, we show that, it represents the solution of a linear system of equations in $n$ unknown where the $i$-th equation is constructed by cases on $1\leq i \leq n$ as:
\begin{description}
\item[if $i 2^n \leq \gamma$)] we construct $\M_i$ as the disjoint union of $\M_\mathcal{G}$ and $\M_V[i 2^n / \gamma]$\footnote{When $\M$ is not a gadget, its rescaling $\M[\varepsilon]$ is obtained by first swapping the start-state (resp.\ end-state) of $\M$ with the source (resp.\ target) link point of the rescaling, then gluing as usual.} where $s$ and $s'$ are their respective start-points. Then the following hold
\begin{align*}
2 \cdot \dist^{\M_i}(s,s') &= \textstyle \sum_{E \in \mathcal{W}_\M} |\Pr{s}(E) - \Pr{s'}(E) | \\
& \textstyle = \Pr{s'}(\alpha\beta) + \sum_{w \in \alpha V^* \omega } |\Pr[\M_\mathcal{G}]{s}(w) - \Pr{s'}(w) | \\
& =  \left( 1 - \frac{i 2^n}{\gamma} \right) + \sum_{w \in \alpha V^* \omega } \left|\Pr[\M_\mathcal{G}]{s}(w) - \frac{i 2^n}{\gamma} \frac{1}{2^n} \right| \\
& =   \left( 1 - \frac{i 2^n}{\gamma} \right) + \frac{1}{\gamma} \sum_{j = 1}^{n} x_j | j - i |
\end{align*}

\item[if $i 2^n > \gamma$)] we construct an SMM $\M_i$ as the disjoint union of $\M_\mathcal{G}[\gamma / i 2^n]$ and $\M_V$ where $s$ and $s'$ are their respective start-points. Then the following hold
\begin{align*}
2 \cdot \dist^{\M_i}(s,s') &= \textstyle \sum_{E \in \mathcal{W}_\M} |\Pr{s}(E) - \Pr{s'}(E) | \\
& \textstyle = \Pr{s}(\alpha\beta) + \sum_{w \in \alpha V^* \omega } |\Pr{s}(w) - \Pr[\M_V]{s'}(w) | \\
& =  \left( 1 - \frac{\gamma}{i 2^n} \right) + \sum_{w \in \alpha V^* \omega } \left| \frac{\gamma}{i 2^n} \Pr[\M_\mathcal{G}]{s}(w) - \frac{1}{2^n} \right| \\
& =   \left( 1 - \frac{\gamma}{i 2^n} \right) + \frac{1}{2^n} \sum_{j = 1}^{n} x_j | j - i |
\end{align*}
\end{description}
In summary, we obtain the linear system $\vec{b} = A \vec{x}$ with equations of the form 
\begin{equation}
\begin{aligned}
&2\gamma \cdot \dist^{\M_i}(s,s') + i 2^n - \gamma = \textstyle\sum_{j = 1}^{n} x_j | j - i |  && \forall 1 \leq i \leq n. \, i 2^n \leq \gamma\\
&2^{n+1} \cdot \dist^{\M_i}(s,s') + \frac{\gamma}{i} - 2^n = \textstyle\sum_{j = 1}^{n} x_j | j - i | && \forall 1 \leq i \leq n. \, i 2^n > \gamma
\end{aligned}
\label{eq:toeplixsystem}
\end{equation}
One can notice that $A$ is an invertible Toeplix matrix. Therefore, provided that one has all the values $\dist^{\M_i}(s,s')$, we can efficiently solve $\vec{b} = A \vec{x}$ in $\Theta(n^2)$. Since the construction of each $\M_i$ is in $O(\mathit{poly}(n))$, we conclude that computing $\dist$ for generic SMMs is \textbf{NP}-hard.
\qed
\end{proof}

\begin{proof}[of Proposition~\ref{prp:inapproximability}]
Let $\mathcal{G}(V,E)$ be an undirected graph with $|V| = n$. In Proof of Theorem~\ref{th:NPhardness} we showed that one can solve Max Clique by solving the Toeplix system $\vec{b} = A \vec{x}$. Specifically, the size of the max clique of $\mathcal{G}$ is $\norm{\vec{x}}$. 

Let define two column vectors $\vec{d},\vec{c} \in \R[]^n$ and a diagonal matrix $D \in \R[]^{n \times n}$ as follows
\begin{align*}
	&d_i = \dist^{\M_i}(s,s')&& 
	c_i = \begin{cases}
		\frac{\gamma - i 2^n}{2 \gamma} & \text{if $i2^n \leq \gamma$} \\
		\frac{1}{2} \left(1 - \frac{\gamma}{i 2^{n+1}} \right) & \text{if $i2^n > \gamma$}
		\end{cases}
	&&
	D_{i,i} = \begin{cases}
		\frac{1}{2 \gamma} & \text{if $i2^n \leq \gamma$} \\
		\frac{1}{2^{n+1}} & \text{if $i2^n > \gamma$}
		\end{cases} \, .
\end{align*} 
The system $\vec{b} = A \vec{x}$ (see \eqref{eq:toeplixsystem} in Proof of Theorem~\ref{th:NPhardness}) can be rewritten as $\vec{d} = D(A\vec{x}) + \vec{c}$. 

Assume that there exists $\vec{d}'$ that over-approximates $\vec{d}$ within an absolute error $\epsilon > 0$, i.e., $\vec{d} \sqsubseteq \vec{d}'$ and $\norm{\vec{d} - \vec{d}'} \leq \epsilon$, and let $\vec{x}'$ be the column vector that solves $\vec{d}' = D(A\vec{x}') + \vec{c}$. Then, the following equalities hold
\begin{align*}
 \vec{d} - \vec{d}' & = D(A\vec{x}) + \vec{c} - D(A\vec{x}') + \vec{c} \\
 	& = D A (\vec{x} - \vec{x}').
\end{align*}
Therefore $\vec{x} - \vec{x}' = A^{-1} D^{-1} (\vec{d} - \vec{d}')$ from which it follows that 
\begin{align*}
  \norm{\vec{x} - \vec{x}'} &= \norm{A^{-1} D^{-1} (\vec{d} - \vec{d}')} \\
  &= \norm{A^{-1}} \norm{D^{-1}} \norm{\vec{d} - \vec{d}'} 
    \tag{sub-multiplicativity} \\
  	& \leq \norm{A^{-1}} \norm{D^{-1}} \epsilon \,.
    \tag{by $\norm{\vec{d} - \vec{d}'} \leq \epsilon$}
\end{align*}
Since $D$ is a diagonal matrix, its inverse is such that $D^{-1}_{i,i} = 1 / D_{i,i}$, therefore $\norm{D^{-1}} = n2^n$. Recalling that $A$ is the Toeplix matrix such that $A_{i,j} = |j-i|$, we let the reader verify that $\norm{A^{-1}} = 1 / (n-1)$. This implies that
\begin{align}
	&|x_i - x'_i| \leq\frac{n2^n}{ n-1}  \epsilon &\forall 1\leq  i \leq n \,. 
	\label{eq:abserr}
\end{align}
The inequality \label{eq:abserr}, states that assuming $\norm{\vec{d} - \vec{d}'} \leq \epsilon$ we can obtain $\vec{x}'$ (as the solution of the system $\vec{d}' = D(A\vec{x}') + \vec{c}$) that approximates the size of the max clique of $\mathcal{G}$ within an absolute error $\frac{n2^n}{ n-1} \epsilon$. Therefore, let $x_j$ be the size of the max clique in $\mathcal{G}$, we can enforce $x'_j$ to become an approximation of $x_j$ within a multiplicative factor $\alpha$, that is $x'_j \leq \alpha x_j$ by enforcing $\epsilon$ to hold $\epsilon \leq (\alpha -1) x_j$. Since $x_j \geq 1$, this requirement is implied by $\epsilon \leq (\alpha -1)$. By this and \eqref{eq:abserr} we obtain that 
\begin{equation}
\epsilon \leq \frac{n-1}{n 2^n} (1 - \alpha) \,. 
\label{eq:alphaerr}
\end{equation}
Since the size of the SMMs used in the reduction is $\Theta(n^2)$ we have that an $\epsilon$ that satisfies \eqref{eq:alphaerr} can be described by a function, say $f$, that depends on the size of the model and the given $\alpha$.
\qed
\end{proof}

\begin{proof}[of Lemma~\ref{lem:couplCorrespondence}] 
First, note that $\eta \colon \pathsC{S} \to \paths{S} \times \paths{S}$ is measurable (the proof is easy and we omit it). 
Let $\C = (S', A', \tau', \rho', \ell')$ be a coupling for $\M$. We show that $\Pr[\C]{s,s'}\#\eta \in \coupling{\Pr{s}}{\Pr{s'}}$. To this end it suffices to prove that for every cylinder $C \in \cyl{S, \Sigma_{\R}}$, the following hold:
\begin{align*}
  \Pr[\C]{s,s'}\#\eta(C \times \paths{S}) = \Pr{s}(C)
  &&\text{and}&&
  \Pr[\C]{s,s'}\#\eta(\paths{S} \times C) = \Pr{s'}(C) \,.
\end{align*}
We proceed by induction on the rank $n \geq 0$ of the cylinders $C \in \cyl{S, \Sigma_{\R}}$.
\begin{itemize}[label={}, fullwidth, topsep=0.5ex, noitemsep]
  \item (Base case $n=0$) Let $C = \cyl{s_0}$, then we have that
  \begin{align*}
    \Pr[\C]{s,s'}\#\eta(\cyl{s_0} \times \paths{S}) 
      &= \Pr[\C]{s,s'}(\cyl{ \set{s_0}{} \times S }) \\
      &= \chi_{\set{(s,s')}{}}(\set{s_0}{} \times S) \\
      &= \chi_{\set{s}{}}(s_0) = \Pr{s}(\cyl{s_0}) \,.
  \end{align*}
  \item (Inductive step $n \geq 0$) Let $C = \cyl{ s_0, R_0, \dots, R_{n}, s_{n+1} }$, then we have
  \begin{align*}
    &\Pr[\C]{s,s'}\#\eta(\cyl{ s_0, R_0, \dots, R_{n}, s_{n+1} } \times \paths{S}) = {} \\
      &= \Pr[\C]{s,s'}( \cyl{ \set{s_0}{} \times S, R_0 \times \R, \dots, R_{n} \times \R, \set{s_{n+1}}{} \times S } ) \\
      &= \Pr[\C]{s,s'}\#\eta(\cyl{ s_0, \dots, s_{n} } \times \paths{S}) \cdot \textstyle
        \sum_{u \notin A}P'((s_{n},u), R_n \times \R, \set{s_{n+1}}{} \times S) \\
      &=\Pr{s}(\cyl{ s_0, \dots, s_{n} }) \cdot \textstyle
        \sum_{u \notin A}P'((s_{n},u), R_n \times \R, \set{s_{n+1}}{} \times S) \,. \\
  \intertext{If $s_n \in A$, the summation $\sum_{u \notin A}P'((s_{n},u), R_n \times \R, \set{s_{n+1}}{} \times S)$ = 0, and we are done, otherwise}
      &=\Pr{s}(\cyl{ s_0, \dots, s_{n} }) \cdot \textstyle
        \sum_{u \notin A} \rho'(s_n,u)(R_n \times \R) \cdot \tau'(s_n,u)(\set{s_{n+1}}{} \times S) \\
      &=\Pr{s}(\cyl{ s_0, \dots, s_n }) \cdot \textstyle
        \sum_{u \notin A} \rho(s_n)(R_n) \cdot \tau(s_n)(s_{n+1}) \\
      &=\Pr{s}(\cyl{ s_0, R_0, \dots, R_{n-1}, s_{n} }) \,.
  \end{align*}
  The right marginal follows similarly. \qed
\end{itemize}
\end{proof}

\begin{proof}[of Lemma~\ref{lem:separability}] 
($\Rightarrow$) Assume $\pi \unsep \pi'$ and let $C \in \mathcal{T}_\M$. We prove that $\pi \in C$ implies  $\pi' \in C$. Let $C = \cyl{[s_0]_{\eqlbl}, R_0, \dots,R_{n-1},[s_n]_{\eqlbl}}$ for some $s_i \in S$ and $R_i \in \Sigma_{\R}$ ($i = 0..n$). By $\pi \in C$ we have that, for all $i \leq n$, $\stateat{\pi}{i} \eqlbl s_i$ and, for all $i < n$, $\timeat{\pi}{i} \in R_i$. From $\pi \unsep \pi'$, we have that $\stateat{\pi}{j} \eqlbl \stateat{\pi'}{j}$ and $\timeat{\pi}{j} = \timeat{\pi'}{j}$ for all $j \in \N$. Hence $\pi' \in C$.

($\Leftarrow$) We proceed by contraposition, proving that if $\pi \unsep \pi'$ then there exists a cylinder $C \in \mathcal{T}_\M$ such that $\pi \in C$ and $\pi' \not\in C$. By $\pi \unsep \pi'$ we have that $\stateat{\pi}{k} \not\eqlbl \stateat{\pi'}{k}$ or $\timeat{\pi}{k} \neq \timeat{\pi'}{k}$ for some $k \in \N$.

If $\stateat{\pi}{k} \not\eqlbl \stateat{\pi'}{k}$ holds we pick $C := \cyl{[\stateat{\pi}{0}]_{\eqlbl},\{\timeat{\pi}{0}\},\dots,\{\timeat{\pi}{k-1}\},[\stateat{\pi}{k}]_{\eqlbl}}$; if $\timeat{\pi}{k} \neq \timeat{\pi'}{k}$ holds we pick $C := \cyl{[\stateat{\pi}{0}]_{\eqlbl},\{\timeat{\pi}{0}\},\dots,\{\timeat{\pi}{k}\},[\stateat{\pi}{k+1}]_{\eqlbl}}$. In both cases $C \in \mathcal{T}_\M$, $\pi \in C$, and $\pi' \notin C$.
\qed
\end{proof}

\begin{proposition} \label{prop:measurablesetAt} 
Let $\M = (S, A, \tau, \rho, \ell)$ be an SMM, $E \subseteq \paths{S}$, and define $E@k \subseteq \paths{S}$ by induction on $k \in \N$ as follows 
\begin{align*}
	E@ 0 = E \,, && 
	E@ (k +1) = tl_\M^{-1}(E@ k) \, .
\end{align*}
Then, for all $k \in \N$, $E@k = \set{\pi}{ \tail{\pi}{k} \in E }$, and if $E \in \Sigma_\M$, then $E@ k \in \Sigma_\M$.
\end{proposition}
\begin{proof}
We proceed by induction on $k \geq 0$.
\begin{itemize}[label={}, topsep=0pt, itemsep=0.5ex]
  \item (Base case: $k = 0$) Trivial.
  \item (Inductive step: $k \geq 0$) $E@(k +1) = tl_\M^{-1}(E@k)$. By inductive hypothesis $E@k \in \Sigma_\M$, so that, by Lemma~\ref{lem:tailmeasurable}, $tl_\M^{-1}(E@k) \in \Sigma_\M$. Moreover,
\begin{align*}
  &E@(k+1) 
  = tl_\M^{-1}(E@k) \tag{def. $@$} \\
  &= tl_\M^{-1}( \set{\pi \in \paths{S}}{ \tail{\pi}{k} \in E }) \tag{ind. hp.} \\
  &= \set{\pi \in \paths{S}}{ \tail{\pi}{k+1} \in E } \tag{def. $tl_\M$}
\end{align*}
\end{itemize}
\qed
\end{proof} 

\begin{proof}[of Lemma~\ref{lem:congmeasurable}] \todo{checked}
We characterize $\sep \subseteq \paths{S} \times \paths{S}$ as a countable union of measurable rectangles in $\Sigma_\M \otimes \Sigma_\M$.

\newcommand{\DiffS}{\parensmathoper{\mathit{DiffS}}}
\newcommand{\DiffT}{\parensmathoper{\mathit{DiffT}}}

Let $\DiffS{k}$ and $\DiffT{k}$ be defined, for $k \geq 0$, as follows
\begin{align*}
  \DiffS{k} &:= \textstyle\bigcup_{C \in S/_{\eqlbl}} \cyl{C}@k \times \cyl{S \setminus C}@k \,,
  \\
  \DiffT{k} &:= \textstyle\bigcup_{ t,t' \in \Q}
    \cyl{S, (t,t'), S}@k \times \cyl{S, \R \setminus (t,t'), S}@k \,.
\end{align*}
where, for $E \subseteq \paths{S}$, $E@k$ is defined as in Proposition~\ref{prop:measurablesetAt}. For any equivalence class $C \in S /_{\eqlbl}$, the cylinders $\cyl{C}$ and $\cyl{S \setminus C}$ are in $\Sigma_\M$; the same holds for $\cyl{S, (t,t'), S}$ and $\cyl{S, \R {\setminus} (t,t'), S}$, for all $t, t' \in \Q$. Therefore, by Proposition~\ref{prop:measurablesetAt}, for all $k \in \N$, $\DiffS{k}$ and $\DiffT{k}$ are a countable union of rectangles in $\Sigma_\M \otimes \Sigma_\M$.
Now we show that $\sep = \bigcup_{k \in \N} \big( \DiffS{k} \cup \DiffT{k} \big)$.
\begin{itemize}[label={}, fullwidth, topsep=0ex, itemsep=0.5ex]

\item ($\subseteq$) Let $\pi \sep \pi'$. By definition $\ell^\omega(\pi) \neq \ell^\omega(\pi')$, that is, there exists an index $k \in \N$ such that $\stateat{\pi}{k} \not\eqlbl \stateat{\pi'}{k}$ or $\timeat{\pi}{k} \neq \timeat{\pi'}{k}$. If $\stateat{\pi}{k} \not\eqlbl \stateat{\pi'}{k}$, by Proposition~\ref{prop:measurablesetAt}, $(\pi, \pi') \in \cyl{[\stateat{\pi}{k}]_{\eqlbl}}@k \times \cyl{S \setminus [\stateat{\pi}{k}]_{\eqlbl}}@k$.

 If $\timeat{\pi}{k} \neq \timeat{\pi'}{k}$ holds, let $\varepsilon = |\timeat{\pi}{k} - \timeat{\pi'}{k}|$. Since $\Q$ is dense in $\R$, every open set has nonempty intersection with $\Q$, so that there exist $t \in \Q \cap (\timeat{\pi}{k} - \varepsilon, \timeat{\pi}{k})$ and $t' \in \Q \cap (\timeat{\pi}{k}, \timeat{\pi}{k}+\varepsilon)$. Clearly, $\timeat{\pi}{k} \in (t,t')$ and $\timeat{\pi'}{k} \notin (t,t')$, therefore, by Proposition~\ref{prop:measurablesetAt}, $(\pi, \pi') \in \cyl{S, (t,t'), S}@k \times  \cyl{S, \R {\setminus} (t,t'), S}@k$.

\item ($\supseteq$) Let $(\pi, \pi') \in \bigcup_{k \in \N} \big( \DiffS{k} \cup \DiffT{k} \big)$, by construction there exists $k \in \N$ such that $(\pi, \pi') \in \DiffS{k}$ or $(\pi,\pi') \in \DiffT{k}$.
\begin{align*}
  \DiffS{k}
    &= \textstyle
    \bigcup_{C \in S/_{\eqlbl}} \cyl{C}@k \times \cyl{S \setminus C}@k  \tag{by def.}\\
    &= \textstyle
    \bigcup_{C \in S/_{\eqlbl}} \set{(\pi, \pi')}{ \pi \in \cyl{C}@k \text{ and } \pi' \in \cyl{S \setminus C}@k}
    \\
    &= \set{(\pi, \pi')}{ \stateat{\pi}{k} \not\eqlbl \stateat{\pi'}{k} } \,,
    \tag{by Prop.~\ref{prop:measurablesetAt}}
\end{align*}
\begin{align*}
  \DiffT{k}
    &= \textstyle\bigcup_{ t,t' \in \Q}
    \cyl{S, (t,t'), S}@k \times \cyl{S, \R {\setminus} (t,t'), S}@k \\
    &= \textstyle\bigcup_{ t,t' \in \Q}
    \left\{ (\pi, \pi') \left| 
      \begin{array}{c}
      \pi \in \cyl{S, (t,t'), S}@k \text{ and } \\
      \pi' \in \cyl{S, \R {\setminus} (t,t'), S}@k
      \end{array}\right. \right\}
    \\
    &= \textstyle\bigcup_{ t,t' \in \Q}
    \set{ (\pi, \pi') }{ \timeat{\pi}{k} \in (t,t') \text{ and } \timeat{\pi'}{k} \notin (t,t') }
     \tag{by Prop.~\ref{prop:measurablesetAt}}
    \\
    &\subseteq \set{(\pi, \pi')}{ \timeat{\pi}{k} \neq \timeat{\pi'}{k}} \,.
\end{align*}
Therefore, there exists an index $k \in \N$ such that $\stateat{\pi}{k} \not\eqlbl \stateat{\pi'}{k}$ or $\timeat{\pi}{k} \neq \timeat{\pi'}{k}$. Thus, $\ell^\omega(\pi) \neq \ell^\omega(\pi')$, that is $\pi \sep \pi'$.
\qed
\end{itemize}
\end{proof} 

\begin{proof}[of Lemma~\ref{lem:dualityvariationandcoupling}] 
We first prove that $\sup_{E \in \Sigma_\M} |\Pr{s}(E) - \Pr{s'}(E)|$ is a lower bound for $\set{ \omega(\sep) }{ \omega \in \coupling{\Pr{s}}{\Pr{s'}} }$. Let $\omega \in \coupling{\Pr{s}}{\Pr{s'}}$ and $E \in \Sigma_\M$, then
\begin{align}
  \Pr{s}(E) 
    &= \omega(E \times \paths{S}) \tag{by $\omega \in \coupling{\Pr{s}}{\Pr{s'}}$} \\
    &\geq \omega(\paths{S} \times E \cap {\unsep}) \tag{by Lemma~\ref{lem:separability}} \\
    &= 1 - \omega((\paths{S} \times E)^c \cup {\sep}) \tag{by complement} \\
    &\geq 1 - \omega((\paths{S} \times E)^c) - \omega({\sep}) \tag{sub additivity} \\
    &= \omega(\paths{S} \times E) - \omega({\sep}) \tag{by complement} \\
    &= \Pr{s'}(E) - \omega({\sep}) \,. \tag{by $\omega \in \coupling{\Pr{s}}{\Pr{s'}}$}
\end{align}
Thus, by the generality of $\omega\in \coupling{\Pr{s}}{\Pr{s'}}$ and $E \in \Sigma_\M$, it immediately follows that $\sup_{E \in \Sigma_\M} |\Pr{s}(E) - \Pr{s'}(E)| \leq \min \set{ \omega(\sep) }{ \omega \in \coupling{\Pr{s}}{\Pr{s'}} }$.

Now we prove that there exists an optimal coupling $\omega^* \in \coupling{\Pr{s}}{\Pr{s'}}$ such that $\omega^*(\sep) = \sup_{E \in \Sigma_\M} |\Pr{s}(E) - \Pr{s'}(E)|$. The following proof is a slightly modification of~\cite[Theorem 5.2]{LindvallBook}. If $\Pr{s} = \Pr{s'}$, $\omega^*$ is the measure that assigns $0$ to all measurable sets contained in $\sep$ and $1$ otherwise. 
Let $\Pr{s} \neq \Pr{s'}$. For each $\unsep$-equivalence class $C \in \paths{S} /_{\unsep}$, fix a representative element $C^* \in C$, and define $\psi \colon \paths{S} \to \paths{S} \times \paths{S}$ by $\psi(\pi) = ([\pi]_{\unsep}^*, [\pi]_{\unsep}^*)$.
We prove that $\psi$ is measurable from $(\paths{S}, \Sigma_\M)$ to $(\paths{S}, \Sigma_\M) \otimes (\paths{S}, \Sigma_\M)$. To do this it suffices to show that for all $E, E' \in \Sigma_\M$, $\psi^{-1}(E \times E') \in \Sigma_\M$. This follows by Lemma~\ref{lem:separability}:
\begin{align*}
  \psi^{-1}(E \times E') 
    &= \set{\pi}{ ([\pi]_{\unsep}^*, [\pi]_{\unsep}^*) \in E \times E'}  \tag{by def. $\psi$} \\
    &= \set{ \pi }{ [\pi]_{\unsep}^* \in E \text{ and } [\pi]_{\unsep}^* \in E' }  \\
    &=  \set{ \pi }{ \pi \in E \text{ and } \pi \in E' } 
    \tag{by Lemma~\ref{lem:separability}} \\
    &= E \cap E' \in \Sigma_\M \,.
\end{align*}
Moreover $\psi^{-1}(\unsep) = \paths{S}$, since all timed paths have their $\unsep$-representative. 

Define the measure $\lambda \colon \Sigma_\M \to \R$ as $\lambda = \Pr{s} + \Pr{s'}$, and the functions $g, g' \colon \paths{S} \to \R$ as the Radon-Nikodym derivatives of $\Pr{s}$ and $\Pr{s'}$, respectively, w.r.t.\ $\lambda$, that is, $g = d \Pr{s}/d \lambda$ and $g' = d \Pr{s'}/d \lambda$. Note that $g$ and $g'$ are well defined, since $\Pr{s}$ and $\Pr{s'}$ are both absolutely continuous w.r.t.\ $\lambda$.
Now, denote by $g \wedge g'$ the point-wise meet of $g$ and $g'$ (i.e., $(g \wedge g')(\pi) = \min \set{g(\pi), g'(\pi)}{}$), and define the measures $\mu \colon \Sigma_\M \to [0,1]$ and $\mu^* \colon \Sigma_\M \otimes \Sigma_\M$ as follows
\begin{align*}
  d \mu/d\lambda = g \wedge g' \,, && \mu^* = \mu\#\psi \,.
\end{align*}
Since $\psi^{-1}(\unsep) = \paths{S}$, $\mu^*$ puts all its mass in $\unsep$. Call this mass $\gamma = \mu^*(\unsep)$, and define $\nu, \nu' \colon \Sigma_\M \to \R[]$ and $\omega^* \colon \Sigma_\M \otimes \Sigma_\M \to [0,1]$ as follows
\begin{align*}
  \nu = \Pr{s} - \mu \,, && \nu' = \Pr{s'} - \mu \,, && \omega^* = \frac{\nu \times \nu'}{1 - \gamma} + \mu^*
\end{align*}
Note that, by the assumption that $\Pr{s} \neq \Pr{s'}$, we have $\gamma < 1$, so that $\omega^*$ is well defined and, in particular, $\omega^*(\unsep) = \gamma$. Now we show that $\omega^* \in \coupling{\Pr{s}}{\Pr{s'}}$. Let $E \in \Sigma_\M$, then
\begin{align}
  \omega^*(E \times \paths{S})
    &= \frac{\nu(E) \cdot \nu'(\paths{S})}{1 - \gamma} + \mu^*(E \times \paths{S}) \tag{by def. $\omega^*$} \\
    &= \frac{\nu(E) \cdot (\Pr{s'}(\paths{S}) - \mu(\paths{S}))}{1 - \gamma} + \mu^*(E \times \paths{S}) 
    \tag{by def. $\nu$} \\
    &= \frac{\nu(E) \cdot (1 - \gamma)}{1 - \gamma} + \mu^*(E \times \paths{S}) 
    \tag{by def. $\mu^*$} \\
    &= \Pr{s}(E) - \mu(E) + \mu^*(E \times \paths{S})     \tag{by def. $\nu$} \\
    &= \Pr{s}(E) - \mu(E) + \mu(E)     \tag{by def. $\mu^*$} \\
    &= \Pr{s}(E) \,. \notag
\end{align}
Similarly $\omega^*(\paths{S} \times E) = \Pr{s'}(E)$, so that the marginals are correct and we indeed have a proper coupling. The following shows that $\omega^*$ is optimal
\begin{align}
  \sup_{E \in \Sigma_\M} |\Pr{s}(E) - \Pr{s'}(E)| 
   &= \sup_{E \in \Sigma_\M} | \int_E g \, d\lambda - \int_E g' \, d\lambda | \tag{Radon-Nikodym} \\
   &= \sup_{E \in \Sigma_\M} \int_E | g - g' | d\lambda \tag{linearity} \\
   &= \frac{1}{2} \int_{\paths{S}} | g - g' | d\lambda \tag{Jordan-Hahn decomposition} \\
   &= \frac{1}{2} \cdot 2 \left( 1 - \int_{\paths{S}} g \wedge g' d\lambda \right) \tag{*} \\
   &= 1 - \mu(\paths{S}) \tag{Radon-Nikodym} \\
   &= 1 - \gamma \tag{def. $\gamma$} \\
   &= 1 - \omega^*(\unsep) \tag{def. $\omega^*$ and $(\nu \times \nu')(\unsep) = 0$} \\
   &= \omega^*(\sep) \tag{complement}
\end{align}
where (*) follows since $2 - \int | g - g' | d\lambda = 2 \int g \wedge g' d\lambda$ (this can be understood considering the geometrical interpretation of integral as ``the area below a function''). This ends the proof, showing that the minimum element exists.
\qed
\end{proof}

\begin{proof}[of Lemma~\ref{lemma:CouplingReachability}]
Let the variable $p_{s,s'}$ denote $(\Pr[\C]{s,s'}\#\eta)(\sep)$. Our goal is to compute $p_{s,s'}$ for all $(s,s') \in S'$. Lemma~\ref{lem:congmeasurable} characterizes $\sep$ as the set of all pairs of times paths $(\pi,\pi')$ s.t.\ $\stateat{\pi}{k} \not\eqlbl \stateat{\pi'}{k}$ or $\timeat{\pi}{k} \neq \timeat{\pi'}{k}$ for some $k \in \N$.

Let consider the set of states $B = \set{(u,v) \in S'}{u \not\equiv v}$.
Clearly, if $B$ is not reachable from $(s,s')$ in the underlying graph of $\C$, then $p_{s,s'} = 0$. Note that this is the case when $s \equiv s$ and $(s,s') \in A'$, that is to say $s \equiv s$ and $s,s' \in A$.

Assume $s \not\equiv s'$, clearly we have that $p_{s,s'} = 1$ since $(s,s') \in B$.

Otherwise, if $(s,s') \not\in B \cup A'$, that is to say $s \equiv s'$ and $s,s' \not\in A$, we have that the probability that a coupled is performed from $(s,s')$ at different time points is the probability associated by $\rho'$ to the event $\{(x,y) \in \R \times \R \mid x \neq y\}$, denoted by $\rho'({\neq})$. Therefore we have that the following equation holds
$$p_{s,s'} = \rho'({\neq}) + (1 - \rho'({\neq})) \cdot \textstyle\sum_{(u,v) \in S'} \tau'(s,s')(u,v) \cdot p_{u,v} \,.$$
The equation states that $p_{s,s'}$ is the probability that one of the two following mutually exclusive events occurs:
\begin{enumerate}[topsep=0.5ex, noitemsep,label=\emph{(\roman*)}]
\item a coupled step from $(s,s')$ occurs at different coupled time points; or 
\item the same occurs starting from some other state $(u,v) \in S'$, after moving there from $(s,s')$ at the same coupled time points. \qed
\end{enumerate} 
\end{proof}

\begin{proof}[of Lemma~\ref{lem:post-fixed}]
Assume $\M = (S, A, \tau, \rho, \ell)$ and $\C = (S', A', \tau', \rho', \ell')$. To prove $f^\M \sqsubseteq d$, it suffices to show that $F^\M(d) \sqsubseteq d$. Indeed, by Tarski's fixed point theorem, $f^\M$ is a lower bound of $\set{d}{F^\M(d) \sqsubseteq d}$. Let $s,s' \in S$, then
\begin{enumerate}[topsep=0.5ex, noitemsep]
  \item if $s \not\equiv s'$, then $F^\M(d)(s,s') = 1 = \Gamma^\C(d)(s,s') = d(s,s')$;
  \item if $s \equiv s'$ and $s,s' \in A$ then $F^\M(d)(s,s') = 0 = \Gamma^\C(d)(s,s') = d(s,s')$;
  \item otherwise, note that for any $0 \leq \alpha \leq \alpha' \leq 1$ and
  $0 \leq \beta \leq \beta' \leq 1$ the following hold
  \begin{align*}
  	\alpha + (1- \alpha) \beta 
	&\leq \alpha + (1- \alpha) \beta' \tag{$\beta \leq \beta'$} \\
	&= \beta' - \beta' + \alpha + (1- \alpha) \beta' \notag \\
	&= \beta' - \alpha \beta' - (1- \alpha) \beta' + \alpha + (1- \alpha) \beta' \tag{$0 \leq \alpha \leq 1$} \\
	&= \beta' - \alpha \beta'  + \alpha  
	= \beta' + (1 - \beta') \alpha \notag \\
	&\leq \beta' + (1 - \beta') \alpha' \tag{$\alpha \leq \alpha'$} \\
	&= \alpha' + (1 - \alpha') \beta'. \tag{$0 \leq \beta \leq 1$}
  \end{align*}
 Moreover, the two following inequalities hold:
  \begin{align*}
  \tv{\rho(s)}{\rho(s')} 
  &= \inf \set{ \bar{\rho}(\neq_{\R})}{\bar{\rho} \in \coupling{\rho(s)}{\rho(s')}} 
  \tag{by \cite[Th.~5.2]{LindvallBook}} \\
  &\leq \rho'(s,s')(\neq_{\R}) \,, \tag{$\rho'(s,s') \in \coupling{\rho(s)}{\rho'(s')}$}
  \end{align*}
and
  \begin{align*}
  \mathcal{K}_d(\tau(s),\tau(s')) 
  &= \textstyle
  \min_{\omega \in \coupling{\tau(s)}{\tau(s')}} \textstyle\sum_{u,v \in S} d(u,v) \cdot \omega(u,v) 
  \tag{by def.} \\
  &\leq \textstyle\sum_{u,v \in S} d(u,v) \cdot \tau'(s,s')(u,v) \,. 
  \tag{$\tau'(s,s') \in \coupling{\tau(s)}{\tau(s')}$} 
  \end{align*}
From the above we have 
 \begin{align}
 	& F^\M(d)(s,s') = \notag \\
 	&= \tv{\rho(s)}{\rho(s')} + (1 - \tv{\rho(s)}{\rho(s')}) \, \mathcal{K}_d(\tau(s),\tau(s')) \tag{by def. $F^\M$} \\
 	&\leq \rho'(s,s')({\neq_{\R}}) + (1 - \rho'(s,s')({\neq_{\R}})) \, \textstyle \sum_{u,v \in S} d(u,v) \cdot \tau'(s,s')(u,v) \notag \\
 	&= \Gamma^\C(d)(s,s') \tag{by def. $\Gamma^\C$} \\
 	&= d(s,s'). \tag{by hypothesis}
 \end{align}
\end{enumerate}
This proves that $F^\M(d) \sqsubseteq d$.
\qed
\end{proof} 

\begin{proof}[of Lemma~\ref{lem:coupling-dist}]
We firstly show that, for any fixed $d \colon S \times S \to [0,1]$, there exists a coupling $\C = (S', A', \tau', \rho', \ell')$ for $\M$ such that $\Gamma^\C(d) = F^\M(d)$. We construct  $\C$ as follows: for each $s,s' \in S$ such that $s \eqlbl s'$, $s,s' \not\in A$, applying \cite[Theorem 5.2]{LindvallBook}, we can fix $\rho'(s,s')$ as the coupling in $\coupling{\rho(s)}{\rho(s')}$ such that ${\rho'(s,s')(\neq_{\R})}= \tv{\rho(s)}{\rho(s')}$; and we fix $\tau'(s,s')$ as one vertex of the transportation polytope $\coupling{\tau(s)}{\tau(s')}$ that achieves the value $\mathcal{K}_d(\tau(s),\tau(s'))$. The sets $S'$, $A'$, and the function $\ell'$ are fixed according to Definition~\ref{def:coupling}. For such a coupling $\C$, its easy to verify that $\Gamma^\C(d) = F^\M(d)$ holds.

Let $\mathcal{D}$ be a coupling for $\mathcal{M}$ such that $\Gamma^\mathcal{D}(f^\M) = F^\M(f^\M)$. By definition $F^\M(f^\M) = f^\M$, therefore $f^\M$ is a fixed point for $\Gamma^\mathcal{D}$. By Lemma~\ref{lem:post-fixed}, $f^\M$ is a lower bound of the set of fixed points of $\Gamma^\mathcal{D}$, so that $f^\M = \discr{D}$. By Lemma~\ref{lem:post-fixed}, we have also that, for any coupling $\mathcal{C}$ for $\mathcal{M}$, $f^\M \sqsubseteq \discr{C}$. Therefore, for $D = \set{\discr{C}}{\text{$\mathcal{C}$ coupling  for $\mathcal{M}$}}$, we have that $f^\M \in D$ and $f^\M$ is a lower bound for $D$. Hence, $f^\M = \min D$. Then, the thesis follows by Corollary~\ref{cor:modelstheta}. \qed
\end{proof}

\begin{proof}[of Theorem~\ref{th:fixedbisimdist}]
By Lemma~\ref{lem:coupling-dist} it suffices to show the two points for $f^\M$.
\begin{enumerate}
\item By showing that $F^\M$ is $\omega$-continuous and proving that, if $d \colon S \times S \to [0,1]$ is a pseudometric then is so $F^\M(d)$. 
\item ($\Rightarrow$) We prove that 
$R = \set{ (s,s') }{ f^\M(s,s') = 0 }$ is a 
bisimulation on $\M$. Clearly, $R$ is an equivalence. Assume $(s,s') \in R$. By definition of $F^\M$, 
it holds:
\begin{enumerate}[topsep=0.5ex, noitemsep]
	\item \label{itm:th:fixedbisimdist1} $s \equiv s'$ and $s,s \in A$, or
	\item \label{itm:th:fixedbisimdist2} $s \equiv s'$ and $s,s \notin A$, and
	$\alpha + (1 - \alpha) \mathcal{K}_{f^\M}(\tau(s),\tau(s')) = 0$ where $\alpha = \tv{\rho(s)}{\rho(s')}$.
\end{enumerate}
If \eqref{itm:th:fixedbisimdist2} holds, we have $\tv{\rho(s)} {\rho(s')} = 0$ and $\mathcal{K}_{f^\M}(\tau(s),\tau(s')) = 0$. The total variation distance is a metric in $\Delta(\R)$, thereore $\rho(s) = \rho(s')$. By \cite[Lemma 3.1]{FernsPP04}, $\mathcal{K}_{f^\M}(\tau(s),\tau(s')) = 0$ implies that, for all $C \in S/_R$, $\tau(s)(C) = \tau(s')(C)$. Therefore $\mathrel{R}$ is a bisimulation.

($\Leftarrow$) 
Let $R$ be a bisimulation on $\M$, and define $d_R \colon S \times S \to [0,1]$ by
$d_R(s,s') = 0$ if $(s,s') \in R$ and $d_R(s,s') = 1$ if $(s,s') \notin R$, for all $s,s' \in S$. We prove that $F^\M(d_R) \sqsubseteq d_R$. If $(s,s') \notin R$, then $d_R(s,s') = 1 \geq F^\M(d_R)(s,s')$. If 
$(s,s') \in R$, then $\ell(s) = \ell(s')$ and one of the following holds:
\begin{enumerate}[topsep=0.5ex, noitemsep]
	\item \label{itm:th:fixedbisimdist3} $s,s' \in A$, or
	\item \label{itm:th:fixedbisimdist4} $s,s' \notin A$, $\rho(s) = \rho(s')$ and,
	$\forall C \in S/_{R}. \, \tau(s)(C) = \tau(s')(C)$.
\end{enumerate}

If \eqref{itm:th:fixedbisimdist3} holds, then $F^\M(d_R)(s,s') = 0 = d_R(s,s')$. 

If \eqref{itm:th:fixedbisimdist4} holds, $\tv{\rho(s)}{\rho(s')} = 0$ and by \cite[Lemma 3.1]{FernsPP04} and the fact that, for all $C \in S/_{R}$, $\tau(s)(C) = \tau(s')(C)$, we have $\mathcal{K}_{d}(\tau(s),\tau(s')) = 0$. Therefore, $F^\M(d_R)(s,s') = 0 = d_R(s,s')$. By the generality of the chosen $R$ and Tarski's fixed point theorem, we have that $s \sim_\M s'$ implies $f^\M(s,s') = 0$. 
\qed
\end{enumerate}
\end{proof}

\begin{lemma} \label{lem:uniquefixedpoint}
Let $\mathcal{M}$ be an SMM and let $G^\M \colon \mDom \to \mDom$ be defined by 
$G(d)(s,s') = 0$ if $s \sim_\M s'$, and $G(d)(s,s') = F^\M(d)(s,s')$ otherwise.
Then, $G$ has a unique fixed point, and it corresponds to $f^\M$.
\end{lemma}
\begin{proof}[of Lemma~\ref{lem:uniquefixedpoint}]
Let $\M = (S, A, \tau, \rho, \ell)$. 
The proof follows the same idea of~\cite[Proposition 17 and Corollary 18]{ChenBW12}. 

We first prove that $G$ has a unique fixed point.
Since $F^\M$ is monotone, it can be easily deduced that $G$ is monotone as well. By Tarski's fixed point theorem, $G$ has a least and a greatest fixed point. Therefore it suffices to prove that if $d \sqsubseteq d'$ are both fixed point of $G$ then $d = d'$. Let 
\begin{align*}
m = \textstyle\max_{s,s' \in S} \set{ d'(s,s') - d(s,s')}{}, 
&& s \mathrel{M} s' \iff d'(s,s') - d(s,s') = m \,. 
\end{align*}
We show that $m = 0$, that is $d = d'$. 
Assume $s \mathrel{M} s'$, we distinguish 3 cases:
\begin{enumerate}
\item if $s \sim_\M s'$, then $d'(s,s') - d(s,s') = 0 - 0 = 0$. Note this covers also the case when $s \equiv s'$ and $s,s' \in A$.
  \item if $s \not\equiv s'$, then $d'(s,s') - d(s,s') = F^\M(d')(s,s') - F^\M(d)(s,s') = 1 - 1 = 0$.
\item otherwise, we have $s \equiv s'$ but $s \not\sim_\M s'$ (note this implies $s,s' \not\in A$). Let $\alpha = \tv{\rho(s)}{\rho(s')}$ and assume that $G(d)(s,s')$ is achieved on $\omega \in \coupling{\tau(s)}{\tau(s')}$, i.e., $G(d)(s,s') = \alpha + (1 - \alpha) \sum_{u,v \in S} d(u,v) \cdot \omega(u,v)$. Then 
\begin{align}
  m &= d'(s,s') - d(s,s') \tag{by $s \mathrel{M} s'$} \\
  &= G(d')(s,s') - G(d)(s,s') \tag{by hp.\ on $d$ and $d'$} \\
  &= (1-\alpha) \big( \mathcal{K}_{d'}(\tau(s),\tau(s')) - \sum_{u,v \in S} d(u,v) \cdot \omega(u,v) \big)\tag{by hp.\ on $\omega$}\\
  &\leq  (1-\alpha) \big(\sum_{u,v \in S} d'(u,v) \cdot \omega(u,v) - \sum_{u,v \in S} d(u,v) \cdot \omega(u,v) \big)
  \tag{by def.\ $\mathcal{K}$} \\
  &=  (1-\alpha) \big(\sum_{u,v \in S} (d'(u,v) - d(u,v)) \cdot \omega(u,v) \big) \notag
\end{align}
By hypothesis on $m$ and $\omega$ we have respectively that $d'(s,s') - d(s,s') \leq m$, and $\sum_{u,v \in S} \omega(u,v) = 1$, thus, from the above inequality, that is 
\begin{equation}
\textstyle m \leq (1-\alpha) \big(\sum_{u,v \in S} (d'(u,v) - d(u,v)) \cdot \omega(u,v) \big) \,,
\label{eq:inequalitym}
\end{equation} we have that 
\begin{itemize}
\item if $0 < \alpha \leq 1$, then $0 \leq 1-\alpha < 0$. Since the left hand side of \eqref{eq:inequalitym} is bounded by $(1-\alpha) m$ we have that \eqref{eq:inequalitym} holds only for $m=0$;
\item if $\alpha = 0$ we have that \eqref{eq:inequalitym} holds only if $d'(u,v) - d(u,v) = m$ whenever $\omega(u,v) > 0$. 
Thus $\omega$ has support contained in $M$. This implies that there exists a coupling model $\C$ for $\M$ (constructed using $\omega$) such that $\discr{C}(s,s') = 0$. By Lemma~\ref{lem:coupling-dist} and Theorem~\ref{th:fixedbisimdist} we have that $s \sim_\M s'$ therefore $m = 0$. 
\end{itemize} 

It remains to prove that $f^\M$ is a fixed point for $G$, that is $f^\M = G(f^\M)$. On the one hand, suppose that $s \sim_\M s'$. Then, by Theorem~\ref{th:fixedbisimdist}, $f^\M(s,s') = 0 = G(f^\M)(s,s')$. On the other hand, if $s \not\sim_\M s'$, $f^\M(s,s') = F(f^\M)(s,s') = G(f^\M)(s,s')$. \qed
\end{enumerate}
\end{proof}

\begin{proof}[of Lemma~\ref{th:bisimpoly}]
In~\cite{ChenBW12} it has been shown that deciding probabilistic bisimilarity over an MC is \textbf{P}-hard. We proceed showing that the problem of computing bisimilarity on an SMM $\M$ can be turned to the  problem of computing probabilistic bisimilarity on an MC $\M'$. 

Recall that the total variation distance is a metric over $\Delta(\R)$, therefore, given two states $s,s' \notin A$, the problem of checking $\rho(s) = \rho(s')$  corresponds to verify $\tv{\rho(s)}{\rho(s')} = 0$. Let $X$ be a set of labels disjoint from $AP$, and let $l \colon S \setminus A \to X$ be a map such that $l(s) = l(s')$ iff $\rho(s) = \rho(s')$. By hypothesis, $l$ can be constructed in polynomial time in $\mathit{size}(\M)$.

Let $\M' = (S,L,\tau',\ell')$ be an MC defined as follows. The set of labels $L$ is $2^{AP \cup X}$; the transition probability $\tau' \colon S \to \Distr{S}$ is defined as $\tau'(s) = \tau(s)$ if $s \not\in A$, and $\tau'(s) = \chi_{\{s\}}$ if $s \in A$; the labeling function $\ell' \colon S \to L$ is defined as $\ell'(s) = \ell(s) \cup \{ l(s) \}$ if $s \not\in A$, and $\ell'(s) = \ell(s)$ if $s \in A$. 

Now we show that any bisimulation $R \subseteq S \times S$ for $\M'$ in the sense of~\cite{ChenBW12} is also a bisimulation for $\M$ in the sense of Definition~\ref{def:pbisim}. Let $R$ be a bisimulation for $\M'$ and assume $s \mathrel{R} s'$. By \cite{ChenBW12} $\ell'(s) = \ell'(s')$ and $\forall C \in S/_{R}. \, \tau'(s)(C) = \tau'(s')(C)$. Clearly, $\ell'(s) = \ell'(s')$ implies that $\ell(s) = \ell(s')$ and either $s,s' \in A$ or $s,s' \notin A$ and $\rho(s) = \rho(s')$. If $s,s' \notin A$, $\forall C \in S/_{R}. \, \tau'(s)(C) = \tau'(s')(C)$ also implies that $\forall C \in S/_{R}. \, \tau(s)(C) = \tau(s')(C)$, by construction of $\tau'$. \qed
\end{proof}

\begin{proof}[of Theorem \ref{th:poly}] 
The proof uses the same idea of~\cite{ChenBW12}. Proceed by showing that $\vartheta^\M$ can be characterized as the solution of a linear program that can be solved in polynomial time in $\mathit{size}(\M)$. 

By Lemmas~\ref{lem:coupling-dist} and \ref{lem:uniquefixedpoint}, $\vartheta^\M$ is the unique fixed point of $G \colon \mDom \to \mDom$. Thus, by Tarski's fixed point theorem we have that 
$$\vartheta^\M = \bigsqcup \set{d \in \mDom }{ d \sqsubseteq G(d)} \,.$$ 
This allows us to characterize $\vartheta^\M$ as the solution of the following linear program 
\begin{align*}
	\text{maximize } & \textstyle\sum_{u,v \in S} d_{u,v} && \\
	\text{such that } & d_{u,v} = 0 && u \sim_\M v \\
				& d_{u,v} = 1 && u \not\equiv v \\
				& \alpha_{u,v} = \tv{\rho(u)}{\rho(v)} && u \equiv v \text{ and } u,v \not\in A \\
				& d_{u,v} \leq \alpha_{u,v} + (1 - \alpha_{u,v}) 
				\sum_{x,y \in S} d_{x,y} \cdot \omega_{x,y}
					&& \begin{aligned}[t]
					& u \equiv v \text{ and } u,v \not\in A, \\
					&\vec{\omega} \in \Omega_{u,v}, \; u \not\sim_\M v
					\end{aligned}
\end{align*}
where, for arbitrary states $s,s' \not\in A$, $\Omega_{s,s'}$ denotes the set of vertices of the transportation polytope $\coupling{\tau(s)}{\tau(s')}$ described by the following linear constraints
\begin{align*}
	&\textstyle\sum_{u \in S} \omega_{u,v} = \tau(s)(u) && u \in S \\
	&\textstyle\sum_{v \in S} \omega_{u,v} = \tau(s')(u) && v \in S \\
	&\omega_{u,v} \geq 0 && u,v \in S
\end{align*}
Indeed, as noticed in \cite{ChenBW12}, for each fixed $\vec{d}$ the linear function mapping a feasible $\vec{\omega}$ to $\sum_{u,v} d_{u,v} \cdot \omega_{u,v}$ achieves its minimum on $\coupling{\tau(s)}{\tau(s')}$ as some vertex. Thus, the (finite) set o vertices $\Omega_{u,v}$ suffices to describe the same feasible region obtained by using the (infinite) set of $\coupling{\tau(s)}{\tau(s')}$.

As showed in \cite{ChenBW12}, even though the above linear program may have exponentially many constraints in the number of states ($\Omega_{u,v}$ may have exponentially many elements), it admits a polynomial-time separation algorithm based on the computation of a transportation problem (see~\cite[Proposition 20]{ChenBW12}) that, given any instance of the variables $d$, it can check for its feasibility in polynomial time and, whether it is not feasible, it returns one of the constraints that is not satisfied. This permits to solve the above linear program in polynomial time in $\mathit{size}(\M)$ using the ellipsoid method.\cite{Schrijver86}.
\qed
\end{proof}

\section{Transportation Problem}

In 1941 Hitchcock and, independently, in 1947 Koopmans considered the problem which 
is usually referred to as the (homogeneous) \emph{transportation problem}. This problem can be 
intuitively described as: a homogeneous product is to be shipped in the amounts $a_1,\dots, a_m$ 
respectively, from each of $m$ shipping \emph{origins} and received in amounts $b_1,\dots,b_n$ 
respectively, by each of $n$ shipping \emph{destinations}. 
The cost of shipping a unit amount from the $i$-th origin to the $j$-th destination is $c_{i,j}$ and 
is known for all combinations $(i,j)$. 
The problem is to determine an optimal \emph{shipping schedule}, i.e.\ the amount $x_{i,j}$ to be 
shipped over all routes $(i,j)$, which minimizes the total cost of transportation.

It can be easily formalized as a linear programming problem
\begin{align*}
\textsf{minimize } & \textstyle\sum_{i = 1}^{m} \sum_{j = 1}^{n} c_{i,j} \cdot x_{i,j} \\
\textsf{such that } 
    & \textstyle\sum_{j=1}^{n} x_{i,j} = a_i && (i = 1,\dots, m) \\
    & \textstyle\sum_{i=1}^{m} x_{i,j} = b_j && (j = 1,\dots, n) \\
    & x_{i,j} \geq 0 && (i = 1,\dots, m \text{ and } j = 1,\dots, n)
\end{align*}
The set of schedules feasible for a transportation problem, which is formalized as
a conjunction of linear constraints, describes a (bounded) convex polytope in 
$\R[2]$, often called transportation polytope. 

There are several algorithms in literature which efficiently solve (not necessarily 
homogeneous) transportation problems. Among these we recall \cite{Dantzig51,FordF56}.

\end{document}